\documentclass[11pt,letterpaper,reqno,final]{article}

\usepackage[english]{babel}
\usepackage{times}
\usepackage{amsfonts,amsmath,amssymb}
\usepackage{amsthm}
\usepackage{url}
\usepackage{ifthen}
\usepackage{ifpdf}
\usepackage{fancyhdr}
\usepackage{float}
\usepackage[letterpaper,hmargin=1in,vmargin=1in]{geometry}

\ifpdf
\usepackage{aeguill}
\usepackage[pdftex,%
pdfauthor={Marcel K. de Carli Silva, Nicholas J.A. Harvey, Cristiane M. Sato},%
pdftitle={Sparse Sums of Positive Semidefinite Matrices},%
letterpaper,%
colorlinks=true,%
hypertexnames=true,%
bookmarks=true,%
pagebackref]{hyperref}
\fi

\newcommand{\newterm}[1]{\textit{#1}}
\newcommand{\comment}[1]{}
\newcommand{\transpose}{^T}

\newcommand{\kaleeps}{\beta}

\theoremstyle{plain}
\newtheorem{theorem}{Theorem}
\newtheorem{proposition}[theorem]{Proposition}
\newtheorem{corollary}[theorem]{Corollary}
\newtheorem{lemma}[theorem]{Lemma}

\newtheorem{definition}[theorem]{Definition}

\newtheorem{problem}[theorem]{Problem}

\newcommand{\repeatclaim}[2]{
    \addvspace{3pt plus 5pt minus 0pt}
    \noindent\textbf{#1. }{\it #2}

    \addvspace{1pt plus 5pt minus 1pt}
}

\newcommand{\repeatclaimcomment}[3]{
    \addvspace{3pt plus 5pt minus 0pt}
    \noindent\textbf{#1 }\textrm{(#2)}\textbf{. }{\it #3}

    \addvspace{1pt plus 5pt minus 1pt}
}

\newenvironment{outer_alg}{
    \begin{list}{}{
        \setlength{\itemsep}{2pt}
        \setlength{\parsep}{0pt}
        \setlength{\parskip}{0pt}
        \setlength{\topsep}{1pt}
        \setlength{\leftmargin}{5pt}
    }
}
{
    \end{list}
}

\newenvironment{alg}{
    \begin{list}{}{
        \setlength{\itemsep}{2pt}
        \setlength{\parsep}{0pt}
        \setlength{\parskip}{0pt}
        \setlength{\topsep}{1pt}
    }
}
{
    \end{list}
}

\floatstyle{ruled}
\newfloat{algorithm}{t}{loa}
\floatname{algorithm}{Algorithm}

\newcommand{\card}[2][]{#1| #2 #1|}
\newcommand{\ceil}[2][]{#1\lceil #2 #1\rceil}
\DeclareMathOperator{\Diag}{Diag}

\newcommand{\drop}{\setminus}
\newcommand{\Ecal}{\mathcal{E}}
\newcommand{\eps}{\varepsilon}
\newcommand{\Fcal}{\mathcal{F}}
\newcommand{\ffrom}{\colon}
\newcommand{\floor}[2][]{#1\lfloor #2 #1\rfloor}
\newcommand{\fto}{\to}
\newcommand{\Gcal}{\mathcal{G}}
\newcommand{\Hcal}{\mathcal{H}}
\newcommand{\iprod}[2]{\langle #1, #2 \rangle}
\newcommand{\iprodt}[2]{\transp{#1}#2}
\newcommand{\Laplacian}[2][]{\Lcal_{\ifthenelse{\equal{#1}{}}{G}{#1}}(#2)}
\newcommand{\Lcal}{\mathcal{L}}
\newcommand{\Lovasz}{Lov\' asz}
\DeclareMathOperator{\mean}{\mathbf{E}}
\newcommand{\myhalf}{\textstyle \frac{1}{2}}
\newcommand{\Ncal}{\mathcal{N}}
\newcommand{\paren}[2][]{#1({#2}#1)}
\DeclareMathOperator{\prob}{\mathbf{P}}
\newcommand{\oprod}[2]{#1#2^{T}}
\newcommand{\oprodsym}[1]{\oprod{#1}{#1}}
\newcommand{\Pcal}{\mathcal{P}}
\newcommand{\Pd}[1][]{\Symraw_{++}^{\ifthenelse{\equal{#1}{}}{n}{#1}}}
\newcommand{\Psd}[1][]{\Symraw_+^{\ifthenelse{\equal{#1}{}}{n}{#1}}}
\newcommand{\qform}[2]{\transp{#2}#1#2}
\newcommand{\Reals}{\mathbb{R}}
\newcommand{\set}[2][]{#1\{ {#2} #1\}}
\newcommand{\setst}[3][]{#1\{\,{#2}\,\colon{#3} #1\}}
\newcommand{\sqbrac}[2][]{#1[{#2}#1]}
\newcommand{\Sym}[1][]{\Symraw^{\ifthenelse{\equal{#1}{}}{n}{#1}}}
\newcommand{\Symraw}{\mathbb{S}}
\DeclareMathOperator{\trace}{Tr}
\newcommand{\transp}[1]{#1^T}
\newcommand{\yb}{\bar{y}}
\newcommand{\zb}{\bar{z}}

\title{Sparse Sums of Positive Semidefinite Matrices}

\author{Marcel K.\ de Carli Silva\thanks{Department of Combinatorics and Optimization,
University of Waterloo. {\tt mksilva@uwaterloo.ca}. Partially supported by an NSERC Discovery Grant of L.~Tun\c{c}el.}
\qquad
Nicholas J. A. Harvey\thanks{Department of Computer Science, University of
British Columbia. {\tt nickhar@cs.ubc.ca}.
Supported by an NSERC Discovery Grant.}
\qquad
Cristiane M.\ Sato\thanks{Department of Combinatorics and Optimization,
University of Waterloo. {\tt cmsato@uwaterloo.ca}.  Partially supported by an NSERC Discovery Grant of N.~Wormald.}
}

\date{}

\begin{document}

\maketitle

\begin{abstract}
    Recently there has been much interest in ``sparsifying''
    sums of rank one matrices:
    modifying the coefficients such that only a few are nonzero,
    while approximately preserving the matrix that results from the sum.
    Results of this sort have found applications in many different areas,
    including sparsifying graphs.
    In this paper we consider the more general problem of sparsifying
    sums of positive semidefinite matrices that have arbitrary rank.

    We give several algorithms for solving this problem.
    The first algorithm is based on the method of Batson, Spielman and Srivastava (2009).
    The second algorithm is based on the matrix multiplicative weights update method of
    Arora and Kale (2007).
    We also highlight an interesting connection between these two algorithms.

    Our algorithms have numerous applications.
    We show how they can be used to construct graph sparsifiers with auxiliary constraints,
    sparsifiers of hypergraphs, and sparse solutions to semidefinite programs.
\end{abstract}

\thispagestyle{empty} \setcounter{page}{0} \clearpage

\section{Introduction}
\label{sec:intro}

A \newterm{sparsifier} of a graph is a subgraph that approximately preserves
some structural properties of the graph.
The original work in this area studied \newterm{cut sparsifiers},
which are weighted subgraphs that approximate every cut arbitrarily well.
The celebrated work of Bencz\'ur and Karger \cite{BKConf,BK} proved that every undirected graph with
$n$ vertices and $m$ edges (and potentially non-negative weights on its edges) has a subgraph with
only $O( n \log n / \eps^2 )$ edges (and new weights on those edges) such that, for
every cut, the weight of the cut in the original graph and its subgraph agree up to a multiplicative
factor of $(1 \pm \eps)$.
 Bencz\'ur and Karger also gave a randomized algorithm to construct a cut sparsifier in
 $\tilde{O}(m/\eps^2)$ time.
 Recent work has extended and improved their algorithm in various ways \cite{FHHP, FH, GKK, HP, HP2}.

Spielman and Teng~\cite{STConf} introduced \newterm{spectral sparsifiers}, which are weighted subgraphs
such that the quadratic forms defined by the Laplacians of the graph and the sparsifier
agree up to a multiplicative factor of $(1 \pm \eps)$.
Spectral sparsifiers are also cut sparsifiers,
as can be seen by evaluating these quadratic forms at $\set{0,1}$-vectors.
They proved that every undirected graph with $n$ vertices and $m$ edges (and potentially
non-negative weights on its edges) has a spectral sparsifier with only
$n \operatorname{polylog}(n) / \eps^2$
edges (and new weights on those edges).
Spielman and Srivastava \cite{SS08} reduce the graph sparsification problem to the following
abstract problem in matrix theory.

\begin{problem}
\label{prob:rank1}
Let $v_1,\ldots,v_m \in \Reals^n$ be vectors and let $B = \sum_i v_i v_i \transpose$.
Given $\eps \in (0,1)$, find a vector $y \in \Reals^m$ with small support such that $y \geq 0$ and
\begin{equation}
\label{eq:sparse-approx-rank1}
B ~\preceq~ \sum_i y_i v_i v_i \transpose ~\preceq~ (1+\eps) B.
\end{equation}
\end{problem}
\noindent (Here the notation $X \preceq Y$ means that the matrix $Y-X$ is
    positive semidefinite.)

Spielman and Srivastava~\cite{SS08} observe that Problem~\ref{prob:rank1} can be solved using
known concentration bounds on operator-valued random variables,
specifically Rudelson's sampling lemma \cite{Rudelson,RV07}.
This approach yields a vector $y$ with support size $O(n \log n / \eps^2)$,
and therefore yields a construction of spectral sparsifiers with $O(n \log n / \eps^2)$ edges.
Their algorithm relies on the linear system solver of Spielman and Teng \cite{STConf},
which was significantly simplified by Koutis, Miller and Peng \cite{KMP}.
Recent work \cite{KL} has improved the space usage of Spielman and Srivastava's algorithm.

In subsequent work, Batson, Spielman and Srivastava \cite{BSS09} give a deterministic algorithm that
solves Problem~\ref{prob:rank1} and produces a vector $y$ with support size $O(n / \eps^2)$.
Consequently they obtain improved spectral sparsifiers with $O(n / \eps^2)$ edges.
This work led to important progress in metric embeddings~\cite{NR,Schechtman},
convex geometry~\cite{Srivastava} and Banach space theory~\cite{SSRI}.

In this paper, we focus on a more general problem.

\begin{problem}
\label{prob:rankn}
Let $B_1,\ldots,B_m$ be symmetric, positive semidefinite matrices of size $n \times n$
and let $B = \sum_i B_i$.
Given $\eps \in (0,1)$, find a vector $y \in \Reals^m$ with small support such that $y \geq 0$ and
\begin{equation}
\label{eq:sparse-approx-rankn}
B ~\preceq~ \sum_i y_i B_i ~\preceq~ (1+\eps) B.
\end{equation}
\end{problem}

This problem can also be solved by known concentration bounds:
Ahlswede and Winter~\cite{AW} give a method for generalizing
Chernoff-like bounds to operator-valued random variables, and one of
their theorems \cite[Theorem 19]{AW} directly yields a solution to
Problem~\ref{prob:rankn}.  (Other expositions of these results also
exist \cite{V,H11}.)  This approach yields a vector $y$ with support
size $O(n \log n / \eps^2)$. See Section~\ref{sec:aw} for more
details.

This paper gives two improved solutions to Problem~\ref{prob:rankn}.
Our interest in this topic is motivated by several applications, such as
constructing sparsifiers with certain auxiliary properties and sparsifiers for hypergraphs.  We discuss
these applications in Section~\ref{sec:app}.

\subsection{Our Results}

We give several efficient algorithms for solving Problem~\ref{prob:rankn}.
Our strongest solution is:

\begin{theorem}
\label{thm:bssextension}
Let $B_1,\ldots,B_m$ be symmetric, positive semidefinite matrices of size $n \times n$
and arbitrary rank.
Set $B := \sum_i B_i$.
For any $\eps \in (0,1)$, there is a deterministic algorithm to construct a vector $y \in \Reals^m$ with
$O(n / \eps^2)$ nonzero entries such that $y \geq 0$ and
$$
B ~\preceq~ \sum_i y_i B_i ~\preceq~ (1+\eps) B.
$$
The algorithm runs in $O(mn^3/\eps^2)$ time. Moreover, the result
continues to hold if the input matrices $B_1, \dotsc, B_m$ are Hermitian
and positive semidefinite.
\end{theorem}

Our proof of Theorem \ref{thm:bssextension} is quite simple and builds
on results of Batson, Spielman and Srivastava \cite{BSS09}. We remark
that the assumption that the $B_i$'s are positive semidefinite cannot be
removed; see Appendix~\ref{sec:psd-assumption}.

\vspace{6pt}

We also give a second solution to Problem~\ref{prob:rankn} which is quantitatively weaker,
although it is based on very general machinery which might prove useful in further
applications or generalizations of Problem~\ref{prob:rankn}.
This second solution is based on the matrix multiplicative weights update method (MMWUM)
of Arora and Kale~\cite{AK07,Kale07}.
By a black-box application of their theorems we obtain a deterministic algorithm
to construct a vector $y$ with $O(n \log n / \eps^3)$ nonzero entries.
By slightly refining their analysis we can improve the number of nonzero entries
to $O(n \log n / \eps^2)$.
We remark that Orecchia and Vishnoi \cite{OV} have used MMWUM for solving the balanced
separator problem; this can be used as a subroutine in Spielman and Teng's algorithm for constructing
spectral sparsifiers.

Another virtue of our second solution is that it illustrates that the
surprising Batson-Spielman-Srivastava (BSS) algorithm
is actually closely related to MMWUM.
In particular, the algorithms underlying our two solutions are \emph{identical},
except for the use of slightly different potential functions.
We explain this connection in Section \ref{sec:bss-connection}.

\subsection{Applications}
\label{sec:app}

In this section, we present several applications of Problem~\ref{prob:rankn}.
Proofs are given in Appendix~\ref{app:app}.

\paragraph{Sparsifiers with costs.}

\newcommand{\corcosts}
{
Let $G = (V,E)$ be a graph, let $w \ffrom E \to \Reals_+$ be a
weight function, and let $c_1, \dotsc, c_k \ffrom E \fto \Reals_+$ be cost functions, with $k = O(n)$.
Let $\Laplacian[G]{w}$ denote the Laplacian matrix for graph $G$ with weight function $w$.
For any real $\eps\in (0,1)$, there is a deterministic polynomial-time algorithm
to find a subgraph $H$ of $G$ and a weight function $w_H\ffrom E(H)\to \Reals_+$ such that
\begin{gather*}
\Laplacian[G]{w} ~\preceq~ \Laplacian[H]{w_H} ~\preceq~ (1+\eps)\Laplacian[G]{w},
\\
\sum_{e\in E} w_e c_{i,e}
    ~\leq~ \sum_{e\in E(H)} w_{H,e} c_{i,e}
    ~\leq~ (1+\eps)\sum_{e\in E} w_e c_{i,e}
    \qquad \text{for all $i$}
\end{gather*}
and $\card{E(H)} = O(n / \eps^2)$.
}

\begin{corollary}
\label{cor:sparsifier-costs}
\corcosts
\end{corollary}

The inequalities
$\Laplacian[G]{w} \preceq \Laplacian[H]{w_H} \preceq (1+\eps)\Laplacian[G]{w}$
are equivalent to the condition that the subgraph $H$
(with weights~$w_H$) is a spectral sparsifier of $G$ (with weights $w$).
We remark that existing methods for producing sparsifiers
have low probability of approximately satisfying even a single cost function
(i.e., the case $k=1$).

One potentially interesting application of sparsifiers with costs is as follows.

\newcommand{\corrainbow}{
    Let $G = (V,E)$ be a graph and
    let $w \ffrom E \to \Reals_+$ be a weight function.
    Let $E_1,\ldots,E_k$ be a partition of the edges,
    i.e., each edge is colored with one of $k$ colors.
    For any real $\eps\in (0,1)$, there is a deterministic polynomial-time algorithm to find a
    subgraph $H$ of $G$ and a weight function $w_H\ffrom E(H)\to \Reals_+$ such that
    \begin{gather*}
    \Laplacian[G]{w} ~\preceq~ \Laplacian[H]{w_H} ~\preceq~ (1+\eps)\Laplacian[G]{w}, \\
    (1-\eps)\sum_{e\in E_i} w_e
        ~\leq~ \sum_{e\in E(H) \cap E_i} w_{H,e}
        ~\leq~ (1+\eps) \sum_{e\in E_i} w_e
        \qquad\text{for all $i$,}
    \end{gather*}
    and $\card{E(H)} = O((n+k) / \eps^2)$.
}

\begin{corollary}[Rainbow Sparsifiers]
\label{cor:rainbow}
\corrainbow
\end{corollary}

\paragraph{Hypergraph sparsifiers.} Let $\Hcal = (V,\Ecal)$ be a
hypergraph, and let $w \ffrom \Ecal \fto \Reals_+$. We follow the
definition of Laplacian for hypergraphs as in~\cite{Rodriguez02}.  For
each hyperedge $E \in \Ecal$, define its Laplacian $\Lcal_E$ as the
graph Laplacian of a graph on~$V$ whose edge set forms a clique
on~$E$. Define the Laplacian for the hypergraph~$\Hcal$ with weight
function~$w$ as the matrix $\Lcal_{\Hcal}(w) := \sum_{E \in \Ecal} w_E
\Lcal_E$.

\newcommand{\corhyp}{
  For any real $\eps\in (0,1)$, there
  is a deterministic polynomial-time algorithm to find a sub-hypergraph
  $\Gcal$ of $\Hcal$ and a weight function $w_{\Gcal}\ffrom
  \Ecal(\Gcal)\to \Reals_+$ such that
    $$
    \Laplacian[\Hcal]{w} ~\preceq~ \Laplacian[\Gcal]{w_{\Gcal}} ~\preceq~ (1+\eps)\Laplacian[\Hcal]{w},
    $$
  and $\card{\Ecal(\Gcal)} = O(n / \eps^2)$.
}

\begin{corollary}[Spectral sparsifiers for hypergraphs]
  \label{cor:hypergraph-sparsifier}
  \corhyp
\end{corollary}

This corollary concerns spectral sparsifiers.
It is also interesting to study sparsifiers that approximately preserve all cuts.
There are several ways to extend the definition of ``the weight of a cut''
from ordinary graphs to hypergraphs.
We consider the following two definitions,
where $S$ is any set of vertices in a hypergraph $\Hcal$ with edge weights $w$.
\begin{itemize}

\item $w(\delta_{\Hcal}(S))$: This is the sum of
the weights of all hyperedges that contain at least one vertex in~$S$ and
at least one vertex in $\overline{S} := V \drop S$.

\item $w^*(\delta_{\Hcal}(S))$: This is defined to be
$ \sum_{E \in \Ecal} w_E \cdot \card{S \cap E} \cdot \card{\overline{S} \cap E}$.
\end{itemize}
Obviously these definitions agree in ordinary graphs.

\newcommand{\corhyptwo}{
  For any real $\eps\in (0,1)$, there
  is a deterministic polynomial-time algorithm to find a sub-hypergraph
  $\Gcal$ of $\Hcal$ and a weight function $w_{\Gcal}\ffrom \Ecal(\Gcal)\to \Reals_+$ such that
    $$
        w^*(\delta_{\Hcal}(S))
        ~\leq~ w_{\Gcal}^*(\delta_{\Gcal}(S))
        ~\leq~ (1+\eps) w^*(\delta_{\Hcal}(S))
        \qquad\text{for every~$S \subseteq V$},
    $$
    and $\card{\Ecal(\Gcal)} = O(n / \eps^2)$.
}

\begin{corollary}[Cut sparsifiers for hypergraphs, second definition]
\label{cor:hyp2}
\corhyptwo
\end{corollary}

\newcommand{\corhypthree}{
  Assume that $\Hcal$ is an $r$-uniform hypergraph.
  For any real $\eps\in (0,1)$, there
  is a deterministic polynomial-time algorithm to find a sub-hypergraph
  $\Gcal$ of $\Hcal$ and a weight function $w_{\Gcal}\ffrom \Ecal(\Gcal)\to \Reals_+$ such that
    $$
        \frac{(r-1)}{r^2/4}w(\delta_{\Hcal}(S))
            ~\leq~ w_{\Gcal}(\delta_{\Gcal}(S))
            ~\leq~ \frac{(1+\eps)r^2}{4(r-1)} w(\delta_{\Hcal}(S))
            \qquad \forall S \subseteq V,
    $$
    and $\card{\Ecal(\Gcal)} = O(n / \eps^2)$.
    In other words, the sparsified hypergraph~$\Gcal$
    approximates the weight of the cuts in the hypergraph~$\Hcal$ to within
    a factor~$\Theta(r^2)$.
}

\begin{corollary}[Cut sparsifiers for hypergraphs, first definition]
\label{cor:hyp3}
\corhypthree
\end{corollary}

For the special case $r = 3$, we can achieve $(1+\eps)$-approximate sparsification
for all cuts, even under the first definition.

\newcommand{\threeunif}{
  Assume that $\Hcal$ is a $3$-uniform hypergraph.
  For any $\eps\in (0,1)$, there
  is a deterministic polynomial-time algorithm to find a sub-hypergraph
  $\Gcal$ of $\Hcal$ and a weight function $w_{\Gcal}\ffrom \Ecal(\Gcal)\to \Reals_+$ such that
  \[
      w(\delta_{\Hcal}(S))
        ~\leq~ w_{\Gcal}(\delta_{\Gcal}(S))
        ~\leq~ (1+\eps) w(\delta_{\Hcal}(S))
        \qquad \forall S \subseteq V,
  \]
  and $\card{\Ecal(\Gcal)} = O(n/\eps^2)$.
}

\begin{corollary}[Cut sparsifiers for 3-uniform hypergraphs]
\label{cor:cut-sparsifiers}
\threeunif
\end{corollary}

\paragraph{Sparse solutions to semidefinite programs.}

\newcommand{\corsdp}
{
  Let $A_1, \dotsc, A_m$ be symmetric, positive semidefinite matrices of
  size $n \times n$, and let $B$ be a symmetric matrix of size $n \times
  n$. Let $c \in \Reals^m$ with $c \geq 0$. Suppose that the
  semidefinite program (SDP)
    $$
    \min\setst[\Big]{\iprodt{c}{z}}{
      \sum_i z_i A_i \succeq B,\, z \in \Reals^m,\, z \geq 0
    }
    $$
  has a feasible solution~$z^*$. Then, for any real $\eps \in (0,1)$, it
  has a feasible solution~$\zb$ with at most $O(n/\eps^2)$ nonzero
  entries and $\iprodt{c}{\zb} \leq (1+\eps)\iprodt{c}{z^*}$.
}

\begin{corollary}
  \label{cor:sparse-sdp}
  \corsdp
\end{corollary}

Several important SDPs can be cast as in Corollary~\ref{cor:sparse-sdp}; see, e.g., \cite{IPS05,IPS11}. Recently, Jain and Yao~\cite{JY11} gave a parallel approximation
algorithm for SDPs in this form with $B$ positive semidefinite.

\paragraph{\Lovasz{} theta number.}
For a graph $G = (V,E)$ on~$n$ nodes, let $t'(G)$ denote the square of
the minimum radius of an Euclidean ball in $\Reals^n$ such that there is
a map from~$V$ to points in the ball such that adjacent vertices are
mapped to points at distance at least $1$. Also, let $\vartheta'(G)$
denote the variant of the \Lovasz{} theta number introduced
in~\cite{McElieceRR78a} and~\cite{Schrijver79a}.

\newcommand{\corhypersphere}{
  Let $G=(V,E)$ be a graph. For any real $\eps \in (0,1)$, there is a
  deterministic polynomial-time algorithm to find a subgraph $H$ of $G$
  such that
  \begin{gather*}
    (1-\eps)t'(G) \leq t'(H) \leq t'(G)
  \end{gather*}
  and $\card{E(H)} = O(n / \eps^2)$.
}

\begin{corollary}
  \label{cor:hypersphere}
  \corhypersphere
\end{corollary}

\newcommand{\corthetap}{
  Let $G=(V,E)$ be a graph. For any real $\eps \in (0,1)$, there is a
  deterministic polynomial-time algorithm to find a supergraph $H$ of
  $G$ such that
  \begin{gather*}
    \frac{\vartheta'(G)}{1-\eps+\eps\vartheta'(G)} \leq \vartheta'(H)
    \leq \vartheta'(G)
  \end{gather*}
  and $\card{E(H)} = \binom{n}{2} - O(n / \eps^2)$.
}

\begin{corollary}
  \label{cor:thetap}
  \corthetap
\end{corollary}

\newcommand{\corthetaplow}{
  Let $G$ be a graph such that $\vartheta'(G) = o(\sqrt{n})$. For any
  real $\gamma >0$, there is a supergraph $H$ of $G$ such that
  \begin{gather*}
    \frac{\vartheta'(G)}{1+\gamma} \leq \vartheta'(H)
    \leq \vartheta'(G)
  \end{gather*}
  and $\card{E(H)} = \binom{n}{2} - O(n\vartheta(G)^2 / \gamma^2)$.
}

\begin{corollary}
  \label{cor:thetaplow}
  \corthetaplow
\end{corollary}

\newcommand{\corthetaphigh}{
  Let $G$ be a graph such that $\vartheta'(G) = \Omega(\sqrt{n})$. For
  any real $\gamma \geq 1$, there is a supergraph $H$ of $G$ such that
  \begin{gather*}
    \vartheta'(H) = \Omega(\sqrt{n}/\gamma)
  \end{gather*}
  and $\card{E(H)} = \binom{n}{2} - O(n^2/\gamma^2)$.
}

\begin{corollary}
  \label{cor:thetaphigh}
  \corthetaphigh
\end{corollary}

\paragraph{Approximate Carath\'eodory theorems.}
One immediate application for Theorem \ref{thm:bssextension} is an
approximate Carath\'eodory-type theorem.
A classic result of this sort is:

\newcommand{\Norm}[1]{\left\lVert #1 \right\rVert}

\begin{theorem}[Alth{\"{o}}fer \cite{Althofer}, Lipton-Young \cite{LY}]
\label{thm:alt}
Let $v_1,\ldots,v_m \in [0,1]^n$ and let $\lambda \in \Reals^m$ satisfy $\lambda \geq 0$
and $\sum_i \lambda_i=1$.
Then there exists $\mu \in \Reals^m$ with $\mu \geq 0$, $\sum_i \mu_i=1$
and only $O(\log n/\eps^2)$ nonzero entries such that
$\Norm{\sum_i \lambda_i v_i - \sum_i \mu_i v_i}_\infty \leq \eps$.
\end{theorem}

This theorem follows from simple random sampling arguments,
but it has several interesting consequences, including
the existence of sparse, low-regret solutions to zero-sum games.
The following corollary of Theorem~\ref{thm:bssextension}
can be viewed as a matrix generalization of Theorem~\ref{thm:alt}.

\newcommand{\corcara}{
Let $B_1,\ldots,B_m$ be symmetric, positive semidefinite matrices of size $n \times n$
and let $\lambda \in \Reals^m$ satisfy $\lambda \geq 0$ and $\sum_i \lambda_i = 1$.
Let $B = \sum_i \lambda_i B_i$.
For any $\eps \in (0,1)$, there exists $\mu \geq 0$ with $\sum_i \mu_i = 1$
such that $\mu$ has $O(n / \eps^2)$ nonzero entries and
$$
(1-\eps)B ~\preceq~ \sum_i \mu_i B_i ~\preceq~ (1+\eps) B.
$$
}
\begin{corollary}
\label{cor:cara}
\corcara
\end{corollary}

Although the support size in Theorem~\ref{thm:alt} is much smaller than in Corollary~\ref{cor:cara},
the latter provides a multiplicative error bound
whereas the former only provides an additive error bound.
Theorem~\ref{thm:alt} can be modified to give multiplicative error bounds
if we allow $\mu$ to have $O(n \log n / \eps^2)$ non-zero entries.
However such a result is not interesting as Carath\'eodory's theorem
provides a $\mu$ with only $n+1$ non-zero entries and no error (i.e., $\epsilon=0$).
In contrast, Carath\'eodory's theorem is very weak in the scenario of Corollary~\ref{cor:cara} as it
only provides a $\mu$ with $n(n+1)/2+1$ nonzero entries.

\paragraph{Sparsifiers on subgraphs.}

\newcommand{\corsubgraph}{
Let $G = (V,E)$ be a graph, let $w \ffrom E \fto \Reals_+$ be a weight function,
and let $\Fcal$ be a collection of subgraphs of~$G$ such that $\sum_{F \in
  \Fcal} \card{V(F)} = O(n)$.
For any real $\eps\in (0,1)$, there is a deterministic polynomial-time algorithm
to find a subgraph $H$ of $G$ and a weight function $w_H\ffrom E(H)\to \Reals_+$
such that $\card{E(H)} = O(n / \eps^2)$ and
\begin{gather*}
    \Laplacian[G]{w} ~\preceq~ \Laplacian[H]{w_H} ~\preceq~ (1+\eps)\Laplacian[G]{w}, \\
    \Laplacian[F]{w_F}
        ~\preceq~ \Laplacian[H \cap F]{w_H\!\!\restriction_{E(H \cap F)}^{}}
        ~\preceq~ (1+\eps)\Laplacian[F]{w_F}
    \qquad\text{for all $F \in \Fcal$},
\end{gather*}
where $w_F := w\!\!\restriction_{E(F)}$ is the restriction of $w$ to the coordinates
$E(F)$ and $H \cap F = \paren[\big]{V(F),E(F) \cap E(H)}$.
}

\begin{corollary}
\label{cor:sim-sparsifier}
\corsubgraph
\end{corollary}

\section{Preliminaries}

For a non-negative integer~$n$, we denote $[n] := \set{1,\dotsc,n}$. The
non-negative reals are denoted by~$\Reals_+$. The set of $n \times n$ symmetric
matrices is denoted by~$\Sym$. The set of symmetric, $n \times n$ positive
semidefinite (resp., positive definite) matrices is denoted by $\Psd$ (resp.,
$\Pd$). Recall that $X \in \Sym$ is positive semidefinite if $\qform{X}{v} \geq
0$ for all $v \in \Reals^n$, and $X$ is positive definite if $X$ is positive
semidefinite and $\qform{X}{v} = 0$ implies $v = 0$. Sometimes we denote $X \in
\Psd$ by $X \succeq 0$ and the notation $X \succeq Y$ means that $X-Y\succeq
0$. For $X \in \Sym$ and $a,b \in \Reals$, the notation $X \in [a,b]$ means that
$aI \preceq X \preceq bI$, where $I$ is the identity matrix.

For $X \in \Sym$, its trace is $\trace X := \sum_{i=1}^n X_{ii}$, its largest (resp., smallest) eigenvalue is denoted by $\lambda_{\max}(X)$ (resp., $\lambda_{\min}(X)$). The vector
space~$\Sym$ can be endowed with the trace inner product
$\iprod{\cdot}{\cdot}$ defined by $\iprod{X}{Y} := \trace(XY) =
\sum_{i,j} X_{ij} Y_{ij}$ for every $X,Y \in \Sym$. We shall repeatedly
use that $\trace(XY) = \trace(YX)$ for any matrices $X,Y$ for which the
products $XY$ and $YX$ make sense.

Let $G = (V,E)$ be a graph.  The canonical basis vectors
of~$\Reals^V$ are $\setst{e_i}{i \in V}$, and the canonical basis vectors
of~$\Reals^E$ are $\setst{e_{\set{i,j}}}{\set{i,j}\in E}$. The Laplacian of~$G$
is the linear transformation $\Laplacian[G]{\cdot} \ffrom \Reals^E \fto
\Sym[V]$ defined by $\Laplacian[G]{w} = \sum_{\set{i,j}\in E} w_{\set{i,j}}
\oprodsym{(e_i-e_j)}$.

When dealing with Problem~\ref{prob:rankn}, we may assume that $B =
I$. See~\cite[Proof of Theorem~1.1]{BSS09} for the details of the reduction.

\section{Solving Problem~\ref{prob:rankn} by Ahlswede-Winter}
\label{sec:aw}

As mentioned earlier, Spielman and Srivastava \cite{SS08} explain how
Problem~\ref{prob:rank1} can be solved by Rudelson's sampling lemma.
This lemma can be easily generalized to handle matrices of arbitrary
rank using the Ahlswede-Winter inequality, yielding a solution to Problem~\ref{prob:rankn}.

Let $X$ be a random matrix such that $X = B_i
/\trace{B_i}$ with probability $p_i := \trace{B_i} / \trace{I}$.
Since $B_i\succeq 0$ and $\sum_i B_i = I$, the $p_i$'s define a
probability distribution.

\begin{theorem}[{\cite[Theorem~19]{AW}}]
  Let $X, X_1, \dotsc, X_T$ be i.i.d.\ random variables with values in
  $\Sym$ such that $X_i \in [0, 1]$ for every $i$ and $\mean(X) = \mu
  I$ with $\mu\in[0,1]$. Let $\eps \in (0,1/2)$. Then
  \begin{equation*}
    \prob \paren[\bigg]{\frac{1}{\mu T} \sum_{i=1}^T X_i \not\in [1-\eps,1+\eps]}
    \leq
    2 n\cdot \exp\left(-T\frac{\eps^2\mu}{2\ln 2}\right).
  \end{equation*}
\end{theorem}

In our case, $\mean(X) =(1/n) I$ and $X \in [0,1]$. So $\mu =
1/n$. Thus, if $T > (2\ln 2)\cdot \frac{\ln n +2 \ln 2}{\eps^2\mu} =
O(n \log n/\eps^2)$, then $\prob \paren[\Big]{\frac{1}{\mu T}
  \sum_{i=1}^T X_i \not\in [1-\eps,1+\eps]} < 1/2$.
Thus, with constant probability, we obtain a solution $y$ to Problem~\ref{prob:rankn}
where $y$ has only $O(n \log n / \eps^2)$ non-zero entries.

\section{Solving Problem~\ref{prob:rankn} by BSS}
\label{sec:bss}

In our modification of the BSS algorithm~\cite{BSS09},
we keep a matrix~$A$ of the form $A = \sum_{i} y_i
B_i$ with $y \geq 0$, starting with $A = 0$, and at each iteration we
add another term $\alpha B_j$ to~$A$. We enforce the invariant that the
eigenvalues of~$A$ lie in $[\ell,u]$, where $u$ and~$\ell$ are parameters given by $u
= u_0 + t\delta_U$ and $\ell = \ell_0 + t\delta_L$ after $t$
iterations.
This procedure is presented in Algorithm~\ref{alg:bss}.
The step of the algorithm which finds $B_j$ and $\alpha$ can be done by exhaustive
search on $j$ and binary search on $\alpha$.
Instead of the binary search, one could also compare the quantities
$U_{A(t-1)}(B_j)$ and $L_{A(t-1)}(B_j)$ defined below.

\begin{algorithm}
\begin{outer_alg}
\item	\textbf{procedure} SparsifySumOfMatricesByBSS($B_1,\ldots,B_m$, $\eps$)
\item	\textbf{input:} Matrices $B_1,\ldots,B_m \in \Psd$ such that $\sum_i B_i = I$, and a parameter $\eps \in (0,1)$.
\item	\textbf{output:} A vector $y$
    with $O(n / \eps^2)$ nonzero entries
    such that $I \preceq \sum_i y_i B_i \preceq (1+O(\eps)) I$.
\item   Initially $A(0):=0$ and $y(0) := 0$.
        Set parameters $u_0, \ell_0, \delta_L, \delta_U$ as in \eqref{eq:bssparams}
        and $T:=4n/\eps^2$.
\item   Define the potential functions
        $\Phi^u(A) := \trace (uI-A)^{-1}$ and $\Phi_{\ell}(A) := \trace (A-\ell I)^{-1}$.
\item   For $t=1,\ldots,T$
    \begin{alg}
        \item Set $u_t := u_{t-1} + \delta_U$ and $\ell_t := \ell_{t-1} + \delta_L$.
        \item Find a matrix $B_j$ and a value $\alpha > 0$ such that
        $A(t-1)+\alpha B_j \in [\ell_t,u_t]$, and
        $$\Phi^{u_t}(A(t-1)+\alpha B_j) \leq \Phi^{u_{t-1}}(A(t-1))
        \quad\text{and}\quad
        \Phi_{\ell_t}(A(t-1)+\alpha B_j) \leq \Phi_{\ell_{t-1}}(A(t-1)).$$
        \item Set $A(t) := A(t-1) + \alpha B_j$ and $y(t) := y(t-1) + \alpha e_j$.
    \end{alg}
\item   Return $y(T) / \lambda_{\min}(A(T))$.
\end{outer_alg}
\caption{A procedure for solving Problem~\ref{prob:rankn} based on the BSS method.}
\label{alg:bss}
\end{algorithm}

In the original BSS algorithm, the matrices are rank one: $B_j = v_j v_j^T$
for some vector~$v_j$.
Their Lemmas~3.3 and~3.4 give sufficient conditions on the new term $\alpha v_j v_j \transpose$
so that the invariant on the eigenvalues is maintained; Lemma~3.5 gives sufficient conditions on
the remaining parameters so that a suitable new term $\alpha v_j v_j^T$ exists
with $\alpha > 0$.
In this section we generalize those lemmas to allow $B_i$ matrices of arbitrary rank.

Let $A \in \Sym$. If $u \in \Reals$ with $\lambda_{\max}(A) < u$, define
$\Phi^u(A) := \trace(uI-A)^{-1}$. If $\ell \in \Reals$ with $\lambda_{\min}(A) >
\ell$, define $\Phi_{\ell}(A) := \trace(A-\ell I)^{-1}$. Note that $\Phi_{\ell}(A) = \sum_i 1/(\lambda_i-\ell)$ and $\Phi^u(A) = \sum_i 1/(u-\lambda_i)$, where $\lambda_1,\dotsc,\lambda_n$ are the eigenvalues of~$A$.
\begin{lemma}[Analog of Lemma 3.3 in \cite{BSS09}]
  \label{lemma:BSS3.3}
  Let $A \in \Sym$ and $X \in \Psd$ with $X \neq 0$. Let $u \in \Reals$ and
  $\delta_U > 0$. Suppose $\lambda_{\max}(A) < u$. Let $u' := u + \delta_U$ and
  $M := u'I-A$. If
  \[
  \frac{1}{\alpha}
  \geq
  \frac{
    \iprod{M^{-2}}{X}
  }{
    \Phi^u(A) - \Phi^{u'}(A)
  }
  +
  \iprod{M^{-1}}{X}
  =: U_A(X),
  \]
  then
  $
  \lambda_{\max}(A+\alpha X) < u'
  $
  and
  $
  \Phi^{u'}(A+\alpha X) \leq \Phi^u(A)$.
\end{lemma}
\begin{proof}
  Clearly $M \succ 0$. Let $V := X^{1/2}$. By the Sherman-Morrison-Woodbury formula \cite{Hager},
  \[
  \begin{split}
    \Phi^{u'}(A+\alpha X)
    & =
    \trace\paren{M-\alpha VV^T}^{-1}
    =
    \trace\paren[\big]{
      M^{-1}
      +
      \alpha M^{-1} V
      \paren{
        I-\alpha\qform{M^{-1}}{V}
      }^{-1}
      V^T M^{-1}
    }
    \\
    & =
    \Phi^{u'}(A)
    +
    \trace\paren[\big]{
      \alpha M^{-1} V
      \paren{
        I-\alpha\qform{M^{-1}}{V}
      }^{-1}
      V^T M^{-1}
    }.
  \end{split}
  \]
  Since $M^{-1} \succ 0$, $X \neq 0$ and $ \Phi^u(A) > \Phi^{u'}(A)$, our
  hypotheses imply
  $1/\alpha
  >
  \iprod{M^{-1}}{X} = \trace\paren{\qform{M^{-1}}{V}}
  \geq
  \lambda_{\max}\paren{\qform{M^{-1}}{V}} \geq 0$,
  so
  $
  \beta :=
  \lambda_{\min}\paren{
    I-\alpha\qform{M^{-1}}{V}
  }
  =
  1
  -
  \alpha\lambda_{\max}\paren{\qform{M^{-1}}{V}}
  > 0
  $
  and by, e.g., \cite[Corollary~7.7.4]{HornJ90a},
  \[
  0
  \prec
  \beta I
  \preceq
  I-\alpha\qform{M^{-1}}{V}
  \implies
  0
  \prec
  \paren{
    I-\alpha\qform{M^{-1}}{V}
  }^{-1}
  \preceq
  \beta^{-1} I.
  \]
  Thus,
  \[
  \begin{split}
    \Phi^{u'}(A+\alpha X)
    & \leq
    \Phi^{u'}(A)
    +
    \alpha\beta^{-1}
    \trace\paren{
      \qform{M^{-2}}{V}
    }
    \\
    & =
    \Phi^u(A)
    -
    \paren{
      \Phi^u(A)
      -
      \Phi^{u'}(A)
    }
    +
    \alpha\beta^{-1}
    \iprod{M^{-2}}{X}
  \end{split}
  \]

  To prove that $\Phi^{u'}(A+\alpha X) \leq \Phi^u(A)$, it suffices to show
  that $ \alpha\beta^{-1} \iprod{M^{-2}}{X} \leq \Phi^u(A) -
  \Phi^{u'}(A)$. This is equivalent to
  \[
  \begin{split}
    &
    \frac{
      \iprod{M^{-2}}{X}
    }{
      1/\alpha
      -
      \lambda_{\max}\paren{\qform{M^{-1}}{V}}
    }
    \leq
    \Phi^u(A)
    -
    \Phi^{u'}(A),
  \end{split}
  \]
  which follows from $1/\alpha \geq U_A(X)$ since $\lambda_{\max}(\qform{M^{-1}}{V})
  \leq \trace(\qform{M^{-1}}{V}) = \iprod{M^{-1}}{X}$.

  It remains to show that $\lambda_{\max}(A+\alpha X) < u'$. Suppose not. Choose $\eps \in (0,\delta_U)$ such that $1/\eps > \Phi^u(A)$.
  By continuity, for some $\alpha' \in (0, \alpha)$ we  have $\lambda_{\max}(A+\alpha'X) = u'-\eps$. Since $1/\alpha' \geq 1/\alpha \geq U_A(X)$, we get
  $\Phi^{u'}(A+\alpha'X) \geq 1/\eps > \Phi^u(A) \geq \Phi^{u'}(A+\alpha'X)$, a contradiction.
\end{proof}

\begin{lemma}[Analog of Lemma 3.4 in \cite{BSS09}]
  \label{lemma:BSS3.4}
  Let $A \in \Sym$ and $X \in \Psd$, with $n \geq 2$. Let $\ell \in \Reals$ and
  $\delta_L > 0$. Suppose $\lambda_{\min}(A) > \ell$ and $\Phi_\ell(A) \leq
  1/\delta_L$. Let $\ell' := \ell + \delta_L$ and $N := A-\ell'I$. If
  \[
  0
  <
  \frac{1}{\alpha}
  \leq
  \frac{
    \iprod{N^{-2}}{X}
  }{
    \Phi_{\ell'}(A) - \Phi_{\ell}(A)
  }
  -
  \iprod{N^{-1}}{X}
  =: L_A(X),
  \]
  then
  $
    \lambda_{\min}(A+\alpha X) > \ell'$ and
    $\Phi_{\ell'}(A+\alpha X) \leq \Phi_{\ell}(A)$.
    Moreover,
    $N \succ 0$.
\end{lemma}
\begin{proof}
  Note that $\lambda_{\min}(A) > \ell$ and $\Phi_{\ell}(A) \leq 1/\delta_L$ imply that
  $N \succ 0$, and therefore $\lambda_{\min}(A+\alpha X) > \ell'$.
  Let $V := X^{1/2}$. By the Sherman-Morrison-Woodbury formula,
  \[
  \begin{split}
    \Phi_{\ell'}(A+\alpha X)
    & =
    \trace\paren{N+\alpha VV^T}^{-1}
    =
    \trace\paren[\big]{
      N^{-1}
      -
      \alpha N^{-1} V
      \paren{
        I+\alpha\qform{N^{-1}}{V}
      }^{-1}
      V^T N^{-1}
    }
    \\
    & =
    \Phi_{\ell'}(A)
    -
    \trace\paren[\big]{
      \alpha N^{-1} V
      \paren{
        I+\alpha\qform{N^{-1}}{V}
      }^{-1}
      V^T N^{-1}
    }.
  \end{split}
  \]

  For
  $
  \beta :=
  \lambda_{\max}\paren{I+\alpha\qform{N^{-1}}{V}},
  $
  we have
  \[
  0
  \prec
  I+\alpha\qform{N^{-1}}{V}
  \preceq
  \beta I
  \implies
  0
  \prec
  \beta^{-1} I
  \preceq
  \paren{
    I+\alpha\qform{N^{-1}}{V}
  }^{-1}.
  \]
  Thus,
  \[
  \begin{split}
    \Phi_{\ell'}(A+\alpha X)
    & \leq
    \Phi_{\ell'}(A)
    -
    \alpha\beta^{-1}
    \trace\paren{
      \qform{N^{-2}}{V}
    }
    \\
    & =
    \Phi_{\ell}(A)
    +
    \paren{
      \Phi_{\ell'}(A)
      -
      \Phi_{\ell}(A)
    }
    -
    \alpha\beta^{-1}
    \iprod{N^{-2}}{X}
  \end{split}
  \]
  We will be done if we show that
  $
  \alpha\beta^{-1}
  \iprod{N^{-2}}{X}
  \geq
  \Phi_{\ell'}(A)
  -
  \Phi_{\ell}(A)$.
  This is equivalent to
  \[
  \begin{split}
    &
    \frac{
      \iprod{N^{-2}}{X}
    }{
      1/\alpha
      +
      \lambda_{\max}\paren{\qform{N^{-1}}{V}}
    }
    \geq
    \Phi_{\ell'}(A)
    -
    \Phi_{\ell}(A)
  \end{split}
  \]
  which follows from $0 < 1/\alpha \leq L_A(X)$, since $\Phi_{\ell'}(A) > \Phi_{\ell}(A)$, $N
  \succ 0$, and $\lambda_{\max}(\qform{N^{-1}}{V}) \leq
  \trace(\qform{N^{-1}}{V}) = \iprod{N^{-1}}{X}$.
\end{proof}

The next lemma can be proved by a syntactic modification of the proof of Lemma~3.5 in~\cite{BSS09}.
\begin{lemma}[Analog of Lemma 3.5 in \cite{BSS09}]
  \label{lemma:BSS3.5}
  Let $A \in \Sym$ with $n \geq 2$, and let $u, \ell \in \Reals$ and $\eps_U,
  \delta_U, \eps_L, \delta_L > 0$ such that $\lambda_{\max}(A) < u$,
  $\lambda_{\min}(A) > \ell$, $\Phi^u(A) \leq \eps_U$, and $\Phi_{\ell}(A) \leq
  \eps_L$. Let $B_1, \dotsc, B_m \in \Sym$ such that $\sum_{i} B_i =
  I$. If
  \begin{equation}
    \label{eq:BSS3.5-condition}
    0 \leq \frac{1}{\delta_U} + \eps_U \leq \frac{1}{\delta_L}
    -\eps_L
  \end{equation}
  then there exists $j \in [m]$ and $\alpha > 0$ for which $L_A(B_j) \geq 1/\alpha \geq U_A(B_j)$.
\end{lemma}
\begin{proof}
  As in~\cite[Lemma~3.5]{BSS09}, it suffices to show that $\sum_i L_A(B_i) \geq \sum_i U_A(B_i)$.
  Let $u' := u + \delta_U$, $M := u'I-A$, $\ell' := \ell + \delta_L$, and $N := A-\ell'I$.
  It follows from the bilinearity of $\iprod{\cdot}{\cdot}$ and the assumption $\sum_i B_i = I$ that
  \begin{subequations}
  \begin{gather}
    \label{eq:BSS3.5-goal1}
    \sum_i U_A(B_i)
    =
    \frac{
      \trace
      M^{-2}
    }{
      \Phi^u(A)-\Phi^{u'}(A)
    }
    +
    \trace M^{-1}
    \\
    \label{eq:BSS3.5-goal2}
    \sum_i L_A(B_i)
    =
    \frac{
      \trace
      N^{-2}
    }{
      \Phi_{\ell'}(A)-\Phi_{\ell}(A)
    }
    -
    \trace N^{-1}
  \end{gather}
  \end{subequations}
It is shown in \cite[Lemma~3.5]{BSS09} that
\eqref{eq:BSS3.5-goal1} is at most \eqref{eq:BSS3.5-goal2},
completing the proof.
\end{proof}

Now we set the parameters of Lemma~\ref{lemma:BSS3.5} similarly as
in~\cite{BSS09}:
\begin{equation}
  \label{eq:bssparams}
  \delta_L := 1
  \qquad
  \eps_L := \frac{\eps}{2}
  \qquad
  \ell_0 := -\frac{n}{\eps_L}
  \qquad
  \delta_U := \frac{2+\eps}{2-\eps}
  \qquad
  \eps_U := \frac{\eps}{2\delta_U}
  \qquad
  u_0 := \frac{n}{\eps_U}.
\end{equation}
So~\eqref{eq:BSS3.5-condition} holds with equality. If $A$ is the matrix
obtained after $T = 4n/\eps^2$ iterations, then
\[
\frac{\lambda_{\max}(A)}{\lambda_{\min}(A)}
\leq
\frac{u_0+T\delta_U}{\ell_0+T\delta_L}
=
\paren[\bigg]{\frac{2+\eps}{2-\eps}}^2
\leq
\frac{1+\eps}{1-\eps}
\]
so $A' := A/\lambda_{\min}(A)$ satisfies $I \preceq A' \preceq
(1+\eps)I/(1-\eps)$ and $A'$ is a positive linear combination of $O(n/\eps^2)$
of the matrices $B_i$.

It is easy to check that the previous lemmas also hold if we replace the
set $\Sym$ of symmetric matrices of size $n \times n$ by the set
$\mathbb{H}^n$ of Hermitian matrices of size $n \times n$.

\subsection{Running Time}

At each iteration, we must compute $U_A(B_j)$ and $L_A(B_j)$ for each $j
\in [m]$. The functions $U_A(X)$ and $L_A(X)$ are the inner products
of~$X$ with certain matrices that can be obtained from~$A$ in time
$O(n^3)$. Thus, each iteration runs in time $O(n^3 + mn^2) = O(mn^2)$,
and the total running time after $T = 4n/\eps^2$ iterations is
$O(mn^3/\eps^2)$. We remark that the reduction to the case $B = I$ can
be made in time $O(mn^3)$. This concludes the proof of
Theorem~\ref{thm:bssextension}.

If the matrices $B_i$ have $O(1)$ nonzero entries, as in the graph
sparsification problem, the algorithm can be made to run in time
$O(n^4/\eps^2+mn/\eps^2)$. We briefly sketch the details. To reduce the
problem to the case that $B = I$, we first compute $(B^+)^{1/2}$, where
$B^+$ is the Moore-Penrose pseudoinverse of~$B$. Define the function
$f(X) := (B^+)^{1/2} X (B^+)^{1/2}$ on~$\Sym$.

The reduction now calls for replacing each input matrix~$B_i$ by
$f(B_i)$ and the matrix~$B$ by~$f(B)$. But we shall not do
this. Instead, we do some preprocessing at each iteration as
follows. The function $U_A(X)$ (as well as~$L_A(X)$) is the inner
product of~$X$ with a certain matrix~$V$. Hence, $U_A(f(B_j)) =
\iprod{V}{f(B_j)} = \iprod{f(V)}{B_j}$ for every~$j$, since~$f$ is
self-adjoint. Thus, to compute $U_A(f(B_j))$ for each~$j$, we first
compute the matrix~$f(V)$ in time $O(n^3)$, and now the inner product
$U_A(f(B_j)) = \iprod{f(V)}{B_j}$ can be computed in constant time for
each~$j$, since $B_j$ has $O(1)$ nonzero entries. Thus, each iteration
runs in time $O(n^3+m)$ and the total running time is
$O(n^4/\eps^2+mn/\eps^2)$.

\section{Solving Problem~\ref{prob:rankn} by MMWUM}
\label{sec:mmwum}

Observe that the set of all vectors $y$ that are feasible for
\eqref{eq:sparse-approx-rankn}
is the feasible region of a semidefinite program (SDP).
So solving Problem~\ref{prob:rankn} amounts to finding a \emph{sparse} solution to this SDP.
Here ``sparse'' means that there are few non-zero entries in the solution $y$;
this differs from other notions of ``low-complexity'' SDP solutions, such as the
low-rank solutions studied by So, Ye and Zhang~\cite{SYZ}.

It has long been known known that the multiplicative weight update method
can be used to construct sparse solutions for some linear programs.
A prominent example is the construction of sparse, low-regret solutions to zero-sum games
\cite{FreundSchapire,YoungGreedy,Y}.
(Another example is the work of Charikar et al.~\cite{CCGGP} on approximating
metrics by few tree metrics.)
Building on that idea, one might imagine that Arora and Kale's
matrix multiplicative update method (MMWUM) \cite{AK07}
can construct sparse solutions to \eqref{eq:sparse-approx-rankn}.
In this section, we show that this is indeed possible:
we obtain a solution $y$ to Problem~\ref{prob:rankn} with $O(n \log n / \eps^3)$
nonzero entries.

\subsection{Overview of MMWUM}

The MMWUM is an algorithm that helps us approximately solve an SDP
feasibility problem. The gist of (a slight modification of) the method
is contained in the following result (its proof can be found in Appendix~\ref{sec:mmwum-proofs}):
\begin{theorem}
  \label{thm:block-mmwum}
  Let $T, K, n_1, \dotsc, n_K$ be positive integers. Let $C_k, A_{1,k},\dotsc,
  A_{m, k} \in \Sym[n_k^{}]$ for $k\in[K]$. For each $k \in [K]$, let $\eta_k >
  0$ and $0 < \beta_k \leq 1/2$. Given~$X_1,\dotsc, X_K \in \Sym$, consider the
  system
  \begin{equation}
    \label{eq:block-oracle}
    \sum_{i=1}^m y_i \iprod{A_{i,k}}{X_k}
    \geq \iprod{C_k}{X_k} - \eta_k \trace X_k,
    \quad
    \forall k\in[K],
    \quad
    \text{and}
    \quad
    y \in \Reals_+^m.
  \end{equation}
  For each $k \in [K]$, let $\set{\Pcal_k,\Ncal_k}$ be a partition of $[T]$, let
  $0 < \ell_k \leq \rho_k$, and let $W_k^{(t)} \in \Sym$ and $\ell_k^{(t)} \in
  \Reals$ for $t \in [T+1]$. Let $y^{(t)} \in \Reals^m$ for $t \in [T]$. Suppose
  the following properties hold: {\allowdisplaybreaks
  \begin{gather*}
    W_k^{(t+1)} = \exp\paren[\bigg]{-\frac{\beta_k}{\ell_k+\rho_k}
      \sum_{\tau=1}^t
      \sqbrac[\bigg]{\sum_{i=1}^m y_i^{(\tau)}
        A_{i,k} - C_k + \ell_k^{(\tau)} I}},
    \qquad
    \forall t \in \set{0,\dotsc,T},\,
    \forall k \in [K],
    \\
    y = y^{(t)} \text{ is a solution for~\eqref{eq:block-oracle} with }
    X_k = W_k^{(t)},\, \forall k \in [K],
    \qquad
    \forall t \in [T],
    \\
    \sum_{i=1}^m y_i^{(t)}
    A_{i,k} - C_k
    \in
    \begin{cases}
      [-\ell_k,\rho_k], & \text{if }t \in \Pcal_k,\\
      [-\rho_k,\ell_k], & \text{if }t \in \Ncal_k,
    \end{cases}
    \qquad
    \forall t \in [T],\, k \in [K],
    \\
    \ell_k^{(t)} = \ell_k^{},
    \quad
    \forall t \in \Pcal_k,\,
    \forall k \in [K],
    \quad\text{and}\quad
    \ell_k^{(t)} = -\ell_k^{},
    \quad \forall t \in \Ncal_k,\,
    \forall k \in [K].
  \end{gather*}
  }
  Define $\yb := \frac{1}{T} \sum_{t=1}^T y^{(t)}$. Then,
  \begin{equation}
    \label{eq:mmwum-conclusion}
    \sum_{i=1}^m \yb_i A_{i,k} - C_k \succeq
    -\sqbrac[\Big]{\beta_k \ell_k + \frac{(\rho_k+\ell_k)\ln n}{T\beta_k}
      + (1+\beta_k)\eta_k}
    I,
    \qquad
    \forall k \in [K].
  \end{equation}
\end{theorem}

Take $K = 2$, set $C_1 := I$ and $C_2 := -I$, and put $A_{i,1} := B_i$ and
$A_{i,2} := -B_i$ for each $i \in [m]$. Then Theorem~\ref{thm:block-mmwum}
shows that finding a solution for~\eqref{eq:sparse-approx-rankn} reduces to
constructing an oracle that solves linear systems of the
form~\eqref{eq:block-oracle} with a few extra technical properties involving the
parameters~$\ell_k$ and~$\rho_k$, and adjusting the other parameters so that the
error term on the right-hand side of~\eqref{eq:mmwum-conclusion} is $\leq \eps$.

To obtain a feasible solution for~\eqref{eq:sparse-approx-rankn} that is also
sparse, the idea is to design an implementation of the oracle that
returns a vector $y^{(t)}$ with only \emph{one} nonzero entry at each
iteration~$t$ of MMWUM, and to adjust the parameters so that, after $T =
O(n \log n/\eps^3)$ iterations, the smallest and largest eigenvalues of
$\sum_{i=1}^m \yb_i B_i$ are $\eps$-close to~$1$. Since $\yb$ is the
average of the $y^{(t)}$'s, the resulting $\yb$ will have at most $T$
nonzero entries.

We set the remaining parameters as follows:
\begin{gather*}
  \beta := \beta_1 := \beta_2 := \frac{\eps}{4},
  \qquad
  T := \frac{2 (\rho +\ell) \ln n}{\beta \eps},
  \qquad
  \eta := \eta_1 := \eta_2 := \frac{\eps}{8},
  \\
  \ell := \ell_1 := \ell_2 := 1,
  \qquad
  \rho := \rho_1 := \rho_2 := \frac{1+\eta}{\eta} n,
  \qquad
  \Pcal_1 := \Ncal_2 := [T],
  \qquad
  \Ncal_1 := \Pcal_2 := \emptyset.
\end{gather*}
Then the error term on the right-hand side of~\eqref{eq:mmwum-conclusion} is
\begin{equation}
  \label{eq:mmwum-small-error}
  \beta \ell + \frac{(\rho + \ell) \ln n}{T \beta} + (1+\beta)\eta
  ~=~
  \frac{\eps}{4}
  +
  \frac{\eps}{2}
  +
  \left( 1 + \frac{\eps}{4}\right) \frac{\eps}{8}
  ~=~
  \frac{7\eps}{8}
  +
  \frac{\eps^2}{32}
  ~\leq~
  \eps.
\end{equation}
Thus, \eqref{eq:sparse-approx-rankn} follows
from~\eqref{eq:mmwum-conclusion}
and~\eqref{eq:mmwum-small-error}. Moreover, $T = O(n\log n/\eps^3)$, as desired.

\subsection{The Oracle}
It remains to implement the oracle. Consider an iteration~$t$, and let
$X_1 := W_1^{(t)}$ and $X_2 := W_2^{(t)}$ be given. We must find $y :=
y^{(t)} \in \Reals_+^m$ with at most one nonzero entry such that
\begin{equation*}
  \begin{split}
    \sum_{i=1}^m y_i \iprod{X_1}{B_i}
    & \geq
    (1-\eta) \trace X_1,
    \qquad
    \sum_{i=1}^m y_i \iprod{X_2}{B_i}
    \leq
    (1 + \eta) \trace X_2,
    \quad
    \text{and}
    \quad
    \sum_{i=1}^m y_i B_i \in [0,\rho].
  \end{split}
\end{equation*}
Since $y$ should have only one nonzero entry, it suffices to find $j\in
[m]$ and $\alpha \in\Reals_+$ such that
\begin{equation}
  \label{eq:our-oracle-simple}
  \begin{split}
    \alpha \iprod{X_1}{B_j}
    & \geq
    (1-\eta) \trace X_1,
    \\
    \alpha \iprod{X_2}{B_j}
    & \leq
    (1+\eta) \trace X_2,
    \\
    \alpha \trace B_j
    & \leq
    \rho.
  \end{split}
\end{equation}
Here we are using the fact that $\lambda_{\max}(B_j) \leq \trace B_j$ since $B_j \succeq 0$.
We will show that such $j$ and $\alpha$ exist.
Due to the definition of $W_1$ and $W_2$, the oracle can assume that $X_1$ is a scalar
multiple of $X_2^{-1}$, although we will not make use of that fact.

\begin{proposition}
  \label{prop:oracle}
  Let $B_1, \dotsc, B_m \in \Psd$ such that $\sum_{i=1}^m B_i = I$. Let
  $\eta > 0$ and $X_1, X_2 \in \Pd$. Then, for $\rho := (1+\eta)n/\eta$,
  there exist $j \in [m]$ and $\alpha \geq 0$ such
  that~\eqref{eq:our-oracle-simple} holds.
\end{proposition}
\begin{proof}
  By possibly dropping some~$B_i$'s, we may assume that $B_i \neq 0$ for
  every $i \in [m]$. Define $p_i := \iprod{X_1}{B_i}/\trace{X_1} > 0$
  for every $i \in [m]$. Consider the probability space on~$[m]$ where
  $j$ is sampled from~$[m]$ with probability~$p_j$. The fact that
  $\sum_{j=1}^m p_j = 1$ follows from $\sum_{i=1}^m B_i = I$. Then
  $\mean_j\sqbrac{p_j^{-1} \trace{B_j}} = \sum_{i=1}^m \trace{B_i} =
  \trace{I} = n$. By Markov's inequality,
  \begin{equation}
    \label{eq:dirty-trick1}
    \prob\paren[\bigg]{p_j^{-1} \trace{B_j} \leq \frac{(1+\eta)}{\eta} n}
    =
    1-
    \prob\paren[\bigg]{p_j^{-1} \trace{B_j} > \frac{(1+\eta)}{\eta} n}
    >
    1- \frac{\eta}{1+\eta}
    =
    \frac{1}{1+\eta}.
  \end{equation}

  Next note that $ \mean_j\sqbrac{p_j^{-1} \iprod{X_2}{B_j}} =
  \sum_{i=1}^m \iprod{X_2}{B_i} = \iprod{X_2}{I} = \trace X_2$.
  Together with Markov's inequality, this yields
  \begin{equation}
    \label{eq:dirty-trick2}
    \prob\paren[\bigg]{p_j^{-1} \iprod{X_2}{B_j} \leq (1+\eta) \trace{X_2}}
    =
    1-
    \prob\paren[\bigg]{p_j^{-1} \iprod{X_2}{B_j} > (1+\eta) \trace{X_2}}
    >
    1- \frac{1}{1+\eta}.
  \end{equation}

  It follows from~\eqref{eq:dirty-trick1} and~\eqref{eq:dirty-trick2}
  that there exists $j \in [m]$ satisfying
  \begin{equation*}
    \begin{split}
      p_j^{-1} \iprod{X_2}{B_j}
      \leq
      (1+\eta) \trace X_2,
      \qquad
      \text{and}
      \qquad
      p_j^{-1} \trace{B_j}
      \leq
      \frac{1+\eta}{\eta}n=
      \rho.
    \end{split}
  \end{equation*}
  Set $\alpha := p_j^{-1}$ and note that
  \begin{equation*}
    \alpha \iprod{X_1}{B_j}
    =
    p_j^{-1} \iprod{X_1}{B_j}
    =
    \trace X_1
    \geq (1-\eta) \trace X_1.
  \end{equation*}
  Hence, $j$ and~$\alpha$ satisfy~\eqref{eq:our-oracle-simple}.
\end{proof}

The following proposition, proven in Appendix~\ref{app:oracle_optimal},
shows that the parameters achieved by Proposition~\ref{prop:oracle}
is essentially optimal.

\newcommand{\proporopt}{
Any oracle for satisfying \eqref{eq:our-oracle-simple}
must have $\rho = \Omega( n / \eta )$,
even if the $B_i$ matrices have rank one,
and even if $X_1$ is a scalar multiple of $X_2^{-1}$.
}
\begin{proposition}
\label{prop:oracle_optimal}
\proporopt
\end{proposition}

We also point out that a naive application of MMWUM as stated by Kale
in~\cite{Kale07} does not work. In his description of MMWUM, the
parameter $K$ is fixed as~$1$. So we must correspondingly adjust our
input matrices to be block-diagonal, e.g., $C$ has two blocks: $I$ and
$-I$. However, applying Theorem~\ref{thm:block-mmwum} in this manner
would lead to a sparsifier with $\Omega(n^2)$ edges.  The reason is that
the parameter $\rho$ needs to be $\Omega(n)$, and we must choose $\ell =
\rho$ since the spectrum of $\sum_{i=1}^m y_i A_i - C$ is symmetric
around zero for any~$y$.  Thus, to get the error term on the right-hand side
of~\eqref{eq:mmwum-conclusion} to be $\leq \eps$, we would need to take $T =
\Omega(n^2)$.

\section{Solving Problem~\ref{prob:rankn} by a Width-Free MMWUM}
\label{sec:tweak}

The algorithm of Section~\ref{sec:mmwum} solves Problem~\ref{prob:rankn} with
only $O(n \log n / \eps^3)$ nonzero entries, which is slightly worse
than the $O(n \log n / \eps^2)$ nonzero entries achieved by the Ahlswede-Winter method
discussed in Section~\ref{sec:aw}.
The main reason for this discrepancy is that MMWUM requires us to bound the ``width''
of the oracle using the parameter $\rho$; formally, the oracle must the
inequality $\alpha \trace B_j \leq \rho$ in \eqref{eq:our-oracle-simple}.
In order to satisfy this width constraint, the oracle loses an extra factor of
$O(1/\eps)$, and this is necessary as shown in Proposition~\ref{prop:oracle_optimal}.

In this section, we slightly refine MMWUM to avoid its dependence on the width.
This allows us to simplify our oracle and avoid losing the extra factor of $O(1/\eps)$.
We obtain a solution to Problem~\ref{prob:rankn} with only
only $O(n \log n / \eps^2)$ nonzero entries,
matching the sparsity of the solutions obtained by the Ahlswede-Winter inequality.

The following theorem is our width-free variant of MMWUM.
We remark that the method described in this theorem is geared towards
solving Problem~\ref{prob:rankn} and is not necessarily useful
for all applications of MMWUM.

\begin{theorem}
  \label{thm:mmwum-bss}
  Let $T$ be a positive integer.
  Let $B_1, \dotsc, B_m \in \Psd$ be nonzero.
  Let $\gamma, \eta, \delta_L, \delta_U > 0$. For any
  given $X_L, X_U \in \Sym$, consider the system
  \begin{equation}
    \label{eq:mmwum-bss-oracle}
    \begin{split}
      & \delta_U \geq
      \frac{\exp(\gamma \alpha \trace B_j)-1}{\trace B_j} \iprod{X_U}{B_j},
      \\
      &
      \delta_L \leq
      \frac{1-\exp(-\gamma \alpha \trace B_j)}{\trace B_j} \iprod{X_L}{B_j},
      \\
      &
      \alpha \in \Reals_+, \quad j \in [m].
    \end{split}
  \end{equation}
  For each $t \in \set{0,\dotsc,T+1}$, let $A(t), W_L(t), W_U(t) \in
  \Sym$, let $\alpha(t) \in \Reals_+$, and let $j(t) \in [m]$. Suppose
  the following properties hold:
  \begin{gather*}
    A(t) = \sum_{\tau=1}^t \alpha(\tau) B_{j(\tau)},
    \qquad
    \forall t \in \set{0,\dotsc,T},
    \\
    W_U(t+1) = \exp(\gamma A(t))
    \qquad
    \text{and}
    \qquad
    W_L(t+1) = \exp(-\gamma A(t)),
    \qquad
    \forall t \in \set{0,\dotsc,T},
    \\
    (\alpha,B_j) = (\alpha(t),B_{j(t)})
    \text{ is a solution for~\eqref{eq:mmwum-bss-oracle} with }
    (X_U,X_L) = \displaystyle \paren[\bigg]{
      \frac{W_U(t)}{\trace W_U(t)},
      \frac{W_L(t)}{\trace W_L(t)}
    },
    \quad
    \forall t \in [T].
  \end{gather*}
  Then
  \begin{equation}
    \label{eq:mmwum-bss-conclusion}
    \frac{A(T)}{T}
    \in
    \sqbrac[\bigg]{
      \frac{\log (1 - \delta_L)^{-1}}{\gamma} - \frac{\log n}{T\gamma},
      \frac{\log (1 + \delta_U)}{\gamma} + \frac{\log n}{T\gamma}
    }.
  \end{equation}
\end{theorem}
\begin{proof}
  We will use Golden-Thompson inequality:
  \begin{equation}
    \label{eq:golden-thompson}
    \trace(\exp(A+B)) \leq \trace(\exp(A)\exp(B)),
    \qquad
    \forall A,B \in \Sym.
  \end{equation}
  We will also make use of the following facts. First,
  $$
  \exp(cx) ~\leq~ 1 + \frac{\exp(c \cdot b)-1}{b} x \qquad \forall c \in
  \Reals,\, b>0,\, x \in [0,b].
  $$
  For $X \in \Psd$, we have $\lambda_{\max}(X) \leq \trace X$, so $X \in
  [0,\trace X]$, and
  \begin{equation}
    \label{eq:expineq}
    \exp(cX) ~\preceq~ I + \frac{\exp(c \cdot \trace X)-1}{\trace X} X.
  \end{equation}

  For each $t \in [T+1]$, define $\Phi_L(t) := \trace W_L(t)$ and
  $\Phi_U(t) = \trace W_U(t)$. For each $t \in [T]$,
  \begin{equation}
    \label{eq:mmuwm-bss-potentials}
    \begin{split}
      \Phi_U(t+1)
      &=
      \trace\paren[\Big]{\exp(\gamma A(t))}
      =
      \trace\paren[\Big]{\exp(\gamma A(t-1) + \gamma \alpha B_{j})}
      \\
      &\stackrel{\eqref{eq:golden-thompson}}{\leq}
      \trace\paren[\Big]{\exp(\gamma A(t-1)) \exp(\gamma \alpha B_{j})}
      \\
      &\stackrel{\eqref{eq:expineq}}{\leq}
      \trace \left(
        \exp(\gamma A(t-1)) \paren[\Big]{\frac{\exp(\gamma\alpha \trace B_{j}) -1}{\trace B_{j}} B_{j} + I}
      \right)
      \\
      &=
      \frac{\exp(\gamma\alpha \trace B_{j}) -1}{\trace B_{j}}
      \trace (\exp(\gamma A(t-1)) B_{j})
      +\trace(\exp(\gamma A(t-1)))
      \\
      &=
      \frac{\exp(\gamma\alpha \trace B_{j}) -1}{\trace B_{j}}
      \iprod{ W_U(t)}{B_{j}}
      + \Phi_U(t)
      \\
      &\stackrel{\eqref{eq:mmwum-bss-oracle}}{\leq}
      (1+\delta_U)\Phi_U(t),
    \end{split}
  \end{equation}
  where we abbreviated $j := j(t)$ and $\alpha := \alpha(t)$.

  Since $A(0) = 0$, we have that $\Phi_U(1) = \trace I =
  n$. Using~\eqref{eq:mmuwm-bss-potentials}, after $T$ iterations,
  \begin{equation*}
    \Phi_U(T+1) \leq (1 + \delta_U)^T n.
  \end{equation*}
  Thus,
  \begin{equation*}
    \exp(\gamma \lambda_{\max}(A(T)))
    \leq
    \sum_{i=1}^n \exp\paren{\gamma \lambda_i}
    = \trace{W_U(T+1)}
    = \Phi_U(T+1) \leq (1 +\delta_U)^T n,
  \end{equation*}
  where $\lambda_1,\dotsc,\lambda_n$ are the eigenvalues of~$A(T)$.
  And so $ \gamma \lambda_{\max}(A(T)) \leq T \log (1 + \delta_U) +
  \log n$, which implies the upper bound
  in~\eqref{eq:mmwum-bss-conclusion}. The proof of the lower bound is
  analogous.
\end{proof}

Next we establish conditions under which we can construct an oracle for solving
the system~\eqref{eq:mmwum-bss-oracle}. The proof consists of
algebraic manipulations and an averaging argument analogous to the proof of Lemma~3.5 in~\cite{BSS09}.
\begin{theorem}
  \label{thm:mmwum-bss-oracle}
  Let $B_1, \dotsc, B_m \in \Psd$ be nonzero such that $\sum_{i=1}^m B_i = I$.
  Let $\delta_U,\delta_L > 0$ be such that
  \begin{equation}
    \label{eq:oracle-deltas}
    \frac{1}{\delta_L}-n
    \geq
    \frac{1}{\delta_U}.
  \end{equation}
  Then, for any $X_L, X_U \in \Pd$ with trace one, the
  system~\eqref{eq:mmwum-bss-oracle} has a solution.
\end{theorem}
\begin{proof}
  The first inequality in~\eqref{eq:mmwum-bss-oracle} is equivalent to
  \begin{equation}
    \label{eq:mmwum-bss-oracle1}
    \frac{\trace B_{j}}{\exp(\gamma\alpha \trace B_{j}) -1} \geq \frac{\iprod{X_U}{B_{j}}}{\delta_U}.
  \end{equation}
  Using the identity $ \frac{1}{1-1/x} = 1+\frac{1}{x-1}$, the second
  inequality in~\eqref{eq:mmwum-bss-oracle} is equivalent to
  \begin{equation}
    \label{eq:mmwum-bss-oracle2}
    \frac{\trace B_{j}}{\exp(\gamma\alpha \trace B_{j}) -1}
    \leq \frac{\iprod{X_L}{B_{j}}}{\delta_L} - \trace B_{j}.
  \end{equation}

  We will choose $j \in [m]$ so that
  \begin{equation}
    \label{eq:goalv}
    \frac{\iprod{X_L}{B_{j}}}{\delta_L} - \trace B_{j}
    \geq
    \frac{\iprod{X_U}{B_{j}}}{\delta_U}
  \end{equation}
  and set $\alpha$ so that~\eqref{eq:mmwum-bss-oracle1} holds with
  equality. Then both~\eqref{eq:mmwum-bss-oracle1}
  and~\eqref{eq:mmwum-bss-oracle2} will hold.  Note that $\alpha \geq 0$
  since $e^{\gamma\alpha \trace B_{j}} = 1 + \delta_U \trace B_{j} / \iprod{X_U}{B_{j}} > 1$
  and $\gamma \trace B_{j} > 0$.

  To see that there exists $j \in [m]$
  satisfying~\eqref{eq:goalv}, note that,
  by~\eqref{eq:oracle-deltas} and $\sum_{i=1}^m B_i = I$,
  \begin{equation*}
    \begin{split}
      \sum_{i=1}^m\left[ \frac{\iprod{X_L}{B_i}}{\delta_L}-\trace B_i \right]
      ~=~
      \frac{\trace{X_L}}{\delta_L}-n
      ~=~
      \frac{1}{\delta_L}-n
      ~\geq~
      \frac{1}{\delta_U}
      ~=~
      \frac{\trace{X_U}}{\delta_U}
      ~=~
      \sum_{i=1}^m \frac{\iprod{X_U}{B_i}}{\delta_U}.
    \end{split}
  \end{equation*}
  This concludes the proof.
\end{proof}

Finally, let us show how to set the parameters to get a sparsifier.
Given $\eps \in (0,1)$, set
\begin{equation}
\label{eq:mmwumparams}
  \eta := \eps/2,
  \qquad
  \delta_U := \frac{\eta}{n},
  \qquad
  \delta_L := \frac{\eta}{(1+\eta)n},
  \qquad
  T := \frac{n \log n}{\eta^2}.
\end{equation}
By our choice of $\delta_L$ and $\delta_U$, we have $1/\delta_L-n =
(1+\eta)n/\eta-n = n/\eta = 1/\delta_U$, so~\eqref{eq:oracle-deltas}
holds with equality. After we run the modified version of MMWUM given by
Theorem~\ref{thm:mmwum-bss}, we obtain a matrix~$A(T)$. Set $\bar{A} := A(T)/T$.
By Theorem~\ref{thm:mmwum-bss},
\begin{equation*}
  \lambda_{\max}(\bar{A})
  \leq
  \frac{\log (1 + \delta_U)}{\gamma} + \frac{\log n}{T\gamma}
  \leq
  \Big( \delta_U + \frac{\eta^2}{n} \Big) / \gamma
  =
  \frac{ 1 + \eta }{ n \gamma / \eta }.
\end{equation*}
We will use that $- \log (1-x)\geq x$ for
$x<1$. Thus,
\begin{equation*}
  \begin{split}
    \lambda_{\min}(\bar{A})
    &\geq
    \frac{\log (1 - \delta_L)^{-1}}{\gamma} - \frac{\log n}{T\gamma}
    \geq
    \Big(\delta_L - \frac{\eta^2}{n} \Big) / \gamma
    =
    \frac{ 1/(1+\eta) - \eta }{ n \gamma/\eta }
    \geq
    \frac{ 1 - 2\eta }{ n \gamma/\eta }.
  \end{split}
\end{equation*}
So if we choose $\gamma = \eta/n$ then
$(1-\eps)I \preceq \bar{A} \preceq (1+\eps)I$ and $\bar{A}$ is of the
form $\sum_i y_i B_i$ with $y \geq 0$ and has at most $T = O(n\log n/\eps^2)$
nonzero entries.

\paragraph{Remark.} The choice of $\gamma$ is actually irrelevant here.
We could choose $\gamma > 0$ arbitrarily, then
define $\bar{A} = A(T) \cdot (n \gamma /\eta T)$ and the desired conclusion would hold.

\section{Solving Problem~\ref{prob:rankn} by Pessimistic Estimators}
\label{sec:pe}

An anonymous reviewer for a preliminary draft of this paper raised the possibility
of designing another deterministic solution to Problem~\ref{prob:rankn}.
The proposal was to use the pessimistic estimators of Wigderson and Xiao~\cite{WX}
to derandomize the random sampling approach of Section~\ref{sec:aw}.
In this section we show that this proposal indeed works.
We remark that pessimistic estimators were also used
by Hofmeister and Lefmann~\cite{HL}
to derandomize the proof of Theorem~\ref{thm:alt}.

It is known that there is a close relationship between pessimistic estimators
and multiplicative weight update methods.
(See, for example, the work of Young~\cite{Y}.)
However, the two methods are not identical,
and in particular the algorithm presented in this section is not identical
to either of our algorithms based on MMWUM.
To illustrate one difference, notice that the algorithm in Section~\ref{sec:aw} has the property
that its output vector $y$ has every component $y_i$ equal to an integer multiple of
$n/(T \cdot \trace B_i)$.
The algorithm of this section also has that property as it is a derandomization
of the algorithm in Section~\ref{sec:aw}.
However, the algorithms in Sections~\ref{sec:bss}, \ref{sec:mmwum} and \ref{sec:tweak}
do not have that property.

\begin{definition}[Definition~3.1 in~\cite{WX}]
  Let $\vec X = (X_1,\dotsc, X_T)$ be random variables distributed
  over $[m]$. Let $S$ be an event with $\prob(\vec X \in S) > 0$.  We
  say that $\phi_0, \dotsc, \phi_T$, $\phi_i: [m]^i \to [0, 1]$, are pessimistic estimators
  for $S$ if the following hold.
  \begin{enumerate}
  \item For any $i$ and any fixed $x_1 , \dotsc, x_i\in [m]$, we have
    that
    \begin{equation*}
      \prob_{X_{i+1},\dotsc, X_T}\paren[\Big]{(x_1,\dotsc, x_i, X_{i+1},\dotsc, X_T)\not\in S}
    \leq\phi_i(x_1,\dotsc, x_i).
    \end{equation*}
  \item For any $i$ and any fixed $x_1 , . . . , x_i \in [n]$:
    \begin{equation*}
      \mean_{X_{i+1}}\paren{\phi_{i+1}(x_1,\dotsc, x_i, X_{i+1})} \leq \phi_i(x_1,\dotsc, x_i).
    \end{equation*}
  \end{enumerate}
\end{definition}
Note that the function $\phi_0$ depends on no variables and is
therefore just a scalar in $[0, 1]$. A nice property of this
definition is that it allows compositions very easily. That is, if we
have pessimistic estimators $\phi_0, \dotsc, \phi_T$ and
$\psi_0,\dotsc, \psi_T$ for events $S$ and $S'$, resp., then
$\phi_0+\psi_0, \dotsc, \phi_T+\psi_T$ are pessimistic estimators for
the event $S\cap S'$ (see Lemma 3.3 in~\cite{WX}).

The key point of this method is that, if there are pessimistic
estimators $\phi_0,\dotsc, \phi_T$, such that $\phi_0 < 1$ and each
$\phi_i$ can be computed efficiently, then one can find $(x_1,\dotsc,
x_T)\in S$ efficiently.

Let $X_1,\dotsc, X_T$ be be i.i.d.\ random variables with same
distribution as the random variable $X$ as defined in
Section~\ref{sec:aw}. Wigderson and Xiao~\cite{WX}
considered the event
$$
S_{\geq} = \set{\vec X:
  \frac{1}{T} \sum_{i=1}^T X_i \succeq (1-\eps)\mu I}
$$
and obtained\footnote{There was an factor of $n$ in the
  $\phi_i$ that can be removed.} the following pessimistic estimators:
\begin{align*}
  &\phi_0 = n
  e^{tT(1-\eps)\mu}\left\|{\mean_X\paren[\big]{\exp(-tX)}}\right\|^{T}
  \leq n \exp(-T\eps^2\mu / (2\ln 2));\\
  &
  \phi_i(x_1,\dotsc, x_i) :=
  e^{tT(1-\eps)\mu}
  \trace\paren[\Big]{\exp(-\sum_{i=1}^{j} tx_i)} \cdot
  \left\|\mean_X\paren[\big]{\exp(-tX)}\right\|^{T-i},
\end{align*}
where $t = \log\paren[\Big]{\frac{1-(1-\eps)\mu}{(1-\mu)(1-\eps)}}$.
Similarly, for the event
$S_{\leq} = \set{\vec X: \frac{1}{T} \sum_{i=1}^T X_i \preceq (1+\eps)\mu I},$
one can find the following pessimistic estimators
\begin{align*}
  &\psi_0 = n
  e^{-t'T(1+\eps)\mu}\left\|{\mean_X\paren[\big]{\exp(t'X)}}\right\|^{T}
  \leq n \exp(-T\eps^2\mu / (2\ln 2));\\
  &
  \psi_i(x_1,\dotsc, x_i) :=
  e^{-t'T(1+\eps)\mu}
  \trace\paren[\Big]{\exp(\sum_{j=1}^{i} t'x_j)} \cdot
  \left\|\mean_X\paren[\big]{\exp(t'X)}\right\|^{T-i},
\end{align*}
where $t' =
\log\paren[\Big]{\frac{(1+\eps)(1-\mu)}{1-(1+\eps)\mu}}$. If we choose
$T > (2 \ln 2) \ln (2n)/(\eps^2\mu) = (2\ln 2) n\ln (2n)/\eps^2$, then
$\phi_0+\psi_0 < 1$. Each $\phi_i,\psi_i$ can be computed efficiently
and so one can find in polynomial time $(x_1,\dotsc, x_T)\in
S_{\geq}\cap S_{\leq}$.

\section{Comparing BSS and MMWUM}
\label{sec:bss-connection}

In this section we show a striking similarity between the algorithms
presented in Sections~\ref{sec:bss} and~\ref{sec:tweak}.
The proof of Theorem~\ref{thm:mmwum-bss} defines two
potential functions for each iteration $t$.
\begin{align*}
\Phi_U(t) &~:=~ \trace W_U(t) = \trace \exp( \gamma A(t) )
\\
\Phi_L(t) &~:=~ \trace W_L(t) = \trace \exp( - \gamma A(t) )
\end{align*}
The proof shows that, for the algorithm of Section~\ref{sec:tweak}, the potentials
must change as follows:
\begin{equation}
\label{eq:phichange}
\begin{split}
    \Phi_U(t+1) &~\leq~ (1+\delta_U) \Phi_U(t)
    \qquad\forall t \in \set{0,\ldots,T-1}
    \\
    \Phi_L(t+1) &~\leq~ (1-\delta_L) \Phi_L(t)
    \qquad\forall t \in \set{0,\ldots,T-1}.
\end{split}
\end{equation}
Instead of requiring these potentials to grow and shrink in this way,
we could instead parameterize the potential functions by the iteration number $t$ and
then simply require that the potential do not grow from iteration to iteration.
To formalize this alternative approach, let us define the new potential functions
\begin{align*}
\Psi^u(A) &~:=~ \trace \exp(-uI+\gamma A),
\\
\Psi_{\ell}(A) &~:=~ \trace \exp(\ell I-\gamma A)
\end{align*}
and define the parameters $\Delta_U = \ln(1+\delta_U)$ and
$\Delta_L = \ln\big((1-\delta_L)^{-1}\big)$.

\begin{proposition}
The inequalities in \eqref{eq:phichange} governing the algorithm's change in
potentials are equivalent to inequalities in \eqref{eq:shift2}.
\begin{equation}
\label{eq:shift2}
\begin{split}
\Psi^{(t+1) \Delta_U}(A(t)+\alpha B_j) &~\leq~ \Psi^{t \Delta_U}(A(t))
\\
\Psi_{(t+1) \Delta_L}(A(t)+\alpha B_j) &~\leq~ \Psi_{t \Delta_L}(A(t))
\end{split}
\end{equation}
\end{proposition}
\begin{proof}
Obviously \eqref{eq:phichange} is equivalent to
\begin{align*}
(1 + \delta_U)^{-(t+1)} \cdot \Phi_U(t+1) &~\leq~ (1 + \delta_U)^{-t} \cdot \Phi_U(t)
\qquad\forall t \in \set{0, \ldots, T-1},
\\
(1 - \delta_L)^{-(t+1)} \cdot \Phi_L(t+1) &~\leq~ (1 - \delta_L)^{-t} \cdot \Phi_L(t)
\qquad\forall t \in \set{0, \ldots, T-1}.
\end{align*}
By the definition of $\Phi_U$ and $\Phi_L$, and by properties of the
exponential function, these inequalities are equivalent to
\begin{equation}
\label{eq:shift}
\begin{split}
\trace \exp(-(t+1) \Delta_U I+\gamma A(t+1)) & ~\leq~ \trace \exp(-t \Delta_U I+\gamma A(t)),
\\
\trace \exp((t+1) \Delta_L I-\gamma A(t+1)) & ~\leq~ \trace \exp(t \Delta_L I-\gamma A(t)).
\end{split}
\end{equation}
Writing $A(t+1) = A(t) + \alpha B_j$, these inequalities in \eqref{eq:shift} are equivalent to
\eqref{eq:shift2}.
\end{proof}

Algorithm~\ref{alg:mmwum} gives pseudocode for the algorithm of Section~\ref{sec:tweak},
using the functions $\Psi^u$ and $\Psi_\ell$ to control the change in potentials.

\begin{algorithm}
\begin{outer_alg}
\item	\textbf{procedure} SparsifySumOfMatricesByMMWUM($B_1,\ldots,B_m$, $\eps$)
\item	\textbf{input:} Matrices $B_1,\ldots,B_m \in \Psd$ such that $\sum_i B_i = I$, and a parameter $\eps \in (0,1)$.
\item	\textbf{output:} A vector $y$
    with $O(n \log n / \eps^2)$ nonzero entries
    such that $I \preceq \sum_i y_i B_i \preceq (1+O(\eps)) I$.
\item   Initially $A(0):=0$, and $y(0) := 0$.
        Set parameters
        $$
            u_0:=0,
            \qquad
            \ell_0:=0,
            \qquad
            \Delta_U := \ln(1+\delta_U),
            \qquad
            \Delta_L := \ln\big((1-\delta_L)^{-1}\big),
        $$
        where $\delta_U$, $\delta_L$ and $T$ are as defined in \eqref{eq:mmwumparams}.
\item   Define the potential functions
        $\Psi^u(A) := \trace \exp(-uI+\gamma A)$ and $\Psi_{\ell}(A) := \trace \exp(\ell I-\gamma A)$.
\item   For $t=1,\ldots,T$
    \begin{alg}
        \item Set $u_t := u_{t-1} + \Delta_U$ and $\ell_t := \ell_{t-1} + \Delta_L$.
        \item Find a matrix $B_j$ and a value $\alpha > 0$ such that
        $$\Psi^{u_t}(A(t-1)+\alpha B_j) \leq \Psi^{u_{t-1}}(A(t-1))
        \quad\text{and}\quad
        \Psi_{\ell_t}(A(t-1)+\alpha B_j) \leq \Psi_{\ell_{t-1}}(A(t-1)).$$
        \item Set $A(t) := A(t-1) + \alpha B_j$ and $y(t) := y(t-1) + \alpha e_j$.
    \end{alg}
\item   Return $y(T) / \lambda_{\min}(A(T))$.
\end{outer_alg}
\caption{A procedure for solving Problem~\ref{prob:rankn} based on the
    Width-Free MMWUM method.}
\label{alg:mmwum}
\end{algorithm}

The main point of this section is to observe that Algorithms~\ref{alg:bss} and~\ref{alg:mmwum}
are identical with the exception of different parameters and different potential functions.
We believe that this similarity between these two algorithms is intriguing,
especially since the BSS algorithm has been called ``highly original'' by Naor \cite{N11}.
In retrospect, it would have been perhaps more natural to develop the BSS algorithm
by the following logical progression of ideas: first observe that MMWUM is
useful for giving sparse solutions to SDPs, then design Algorithm~\ref{alg:mmwum},
then later realize that a clever refinement of it leads
to Algorithm~\ref{alg:bss} and its improved analysis.
It is remarkable that Batson, Spielman and Srivastava developed their
algorithm from first principles, apparently without knowing this connection to
established algorithmic techniques.

With the advantage of hindsight (i.e., the knowledge that the BSS algorithm exists),
we now explain how one might be tempted to refine Algorithm~\ref{alg:mmwum}.
It is quite tempting to modify the potential functions to more strongly
penalize eigenvalues which deviate from the desired range.
The natural approach to do this would be to increase the derivatives of the potential
function by increasing the parameter $\gamma$.
However, as remarked at the end of Section~\ref{sec:tweak},
the algorithm is actually unaffected by varying $\gamma$!
Thus, to improve Algorithm~\ref{alg:mmwum},
one must seek a more substantially different potential function.

Focusing on the upper potential, we consider the question:
is there a function $f \ffrom \Reals \fto \Reals$
with steeper derivatives than $\exp(u-x)$ and such that,
for any matrices $A$ and $B$,
$\trace f(A+B)$ can be easily related to $\trace f(A)$?
The natural candidates to try are $f(x) = -\log(u-x)$ and $f(x) = (u-x)^{-1}$
since, in both cases, $\trace f(A+B)$ can be related to $\trace f(A)$
by the Sherman-Morrison-Woodbury formula.
We do not know whether the choice $f(x) = -\log(u-x)$ can be made to work.
However, choosing $f(x) = (u-x)^{-1}$, one arrives at Algorithm~\ref{alg:bss},
our generalization of the BSS algorithm.
Of course, even after arriving at this algorithm, one must also analyze it,
and this requires the delicate calculations
that were accomplished by Batson, Spielman and Srivastava.

\section*{Acknowledgements}

We thank Satyen Kale for helpful discussions.

\newpage
\appendix

\section{Proofs of the Applications}
\label{app:app}

\repeatclaim{Corollary \ref{cor:sparsifier-costs}}{\corcosts}

\begin{proof}
  For every edge $e=ij\in E$, let $B_e$ be the direct sum $
  w_{ij}\sqbrac[\big]{(e_i-e_j)(e_i-e_j)^T \oplus c_{1,e} \oplus\dotsm \oplus
  c_{k,e}}$. Let $B := \Laplacian[G]{w}\oplus \iprodt{w}{c_1}
\oplus\dotsm\oplus \iprodt{w}{c_k}$.  The result
  follows immediately by applying Theorem~\ref{thm:bssextension} to these
  matrices.
\end{proof}

\repeatclaim{Corollary \ref{cor:rainbow}}{\corrainbow}

\begin{proof}
  For each $i$, let $c_i \ffrom E \fto \Reals$ be the
  characteristic vector of~$E_i$. Now apply
  Corollary~\ref{cor:sparsifier-costs}.
\end{proof}

\repeatclaimcomment{Corollary \ref{cor:hypergraph-sparsifier}}
    {Spectral sparsifiers for hypergraphs}
    {\corhyp}

\begin{proof}
  The result follows directly by applying Theorem~\ref{thm:bssextension}
  to the matrices $w_E\Lcal_E$.
\end{proof}

\repeatclaimcomment{Corollary \ref{cor:hyp2}}
    {Cut sparsifiers for hypergraphs, second definition}
    {\corhyptwo}

\begin{proof}
Note that $w^*(\delta_{\Hcal}(S))$ is
obtained by evaluating the quadratic form $\qform{\Lcal_{\Hcal}(w)}{x}$,
where $x$ is the characteristic vector of~$S$.
Thus the sparsifier produced by Corollary~\ref{cor:hypergraph-sparsifier} satisfies
the desired inequalities.
\end{proof}

\repeatclaimcomment{Corollary \ref{cor:hyp3}}
    {Cut sparsifiers for hypergraphs, first definition}
    {\corhypthree}

\begin{proof}
For any $r$-uniform hypergraph $\Hcal$, it is easy to see that
\begin{equation}
  \label{eq:cut-weights}
  (r-1)w(\delta_{\Hcal}(S)) \leq
  w^*(\delta_{\Hcal}(S)) \leq \floor{r/2}\ceil{r/2} w(\delta_{\Hcal}(S))
  \qquad
  \forall S \subseteq V.
\end{equation}
Thus the sparsifier produced by Corollary~\ref{cor:hypergraph-sparsifier} satisfies
the desired inequalities.
\end{proof}

\repeatclaimcomment{Corollary \ref{cor:cut-sparsifiers}}
    {Cut sparsifiers for 3-uniform hypergraphs}
    {\threeunif}

\begin{proof}
  Since $r=3$, a consequence of~\eqref{eq:cut-weights} is that $w^*(\delta_{\Hcal}(S)) =
  2w(\delta_{\Hcal}(S))$ for every~$S$.
Thus the sparsifier produced by Corollary~\ref{cor:hypergraph-sparsifier} satisfies
the desired inequalities.
\end{proof}

\repeatclaim{Corollary \ref{cor:sparse-sdp}}{\corsdp}
\begin{proof}
  Let $B_i' :=
  \begin{bmatrix}
    z_i^*A_i &\mathbf{0}\\
    \mathbf{0} & c_iz_i^*
  \end{bmatrix}$ for every $i\in[m]$ and
  $B' :=
  \begin{bmatrix}
    D& \mathbf{0}\\
    \mathbf{0} & c^Tz^*
  \end{bmatrix}$, where $D := \sum_i z_i^* A_i \succeq B$. Then $B_i'\succeq 0$
  and $B'= \sum_i B_i'$.  By applying Theorem~\ref{thm:bssextension}, we
  obtain $y\in\Reals^m$ with $y\geq 0$ and $O(n/\eps^2)$ nonzero entries
  such that $\sum_i y_i z_i^* A_i \succeq D \succeq B$ and $\sum_i c_i
  y_i z_i^* \leq (1+\eps) c^T z^*$. Thus, we can take $\bar z_i =
  y_iz_i^*$ for every $i\in[m]$.
\end{proof}

\repeatclaim{Corollary \ref{cor:hypersphere}}{\corhypersphere}
\begin{proof}
  It is straightforward to formulate $t'(G)$ as an SDP (see, e.g.,
  \cite{LovaszSDP}) so that its dual has an optimal solution and there
  is no duality gap. The dual can be written as:
  \begin{equation}
    \label{opt:tprime}
    \max
    \setst[\Big]{
      \sum_{e \in E} z_e
    }{
      \Diag(y)\succeq \Laplacian{z},\,
      \sum_{v\in V} y_v = 1,\,
      z \geq 0
    }
  \end{equation}
  The proof is now almost identical to the proof of
  Corollary~\ref{cor:sparse-sdp}. Let $(z^*,y^*)$ be an optimal
  solution. Using Theorem~\ref{thm:bssextension}, we obtain $\bar z\in
  \Reals^E$ with $\bar z\geq 0$ and $O(n/\eps^2)$ nonzero entries such
  that $(y^*,\bar z)$ is feasible in~\eqref{opt:tprime} and has
  objective value $\sum_{e\in E(H)}\bar z_e\geq (1-\eps) t(G)$, where $H
  = (V,E(H))$ and $E(H)$ is the support of $\bar z$. Then $\bar z$ is
  also feasible for the SDP defined using $H$ instead of $G$, which
  shows that $t'(H) \geq (1-\eps)t'(G)$.
\end{proof}

\repeatclaim{Corollary \ref{cor:thetap}}{\corthetap}
\begin{proof}
  For a graph $G = (V,E)$, define $t(G)$ as the square of the minimum
  radius of a hypersphere on $\Reals^n$ such that there is a map
  from~$V$ to the hypersphere such that adjacent vertices are mapped to
  points at distance exactly $1$. \Lovasz~\cite{LovaszSDP} noted that
  $t(G)$ is related to the \Lovasz{} theta number
  $\vartheta(\overline{G})$ of the complement~$\overline{G}$ of~$G$ by
  the formula $2t(G) + 1/\vartheta(\overline{G}) =1$; see~\cite{dCST11a}
  for a proof. By repeating the same proof for $t'(G)$, one finds that
  $2t'(G) + 1/\vartheta'(\overline{G}) =1$. The result now follows from
  Corollary~\ref{cor:hypersphere} via this formula.
\end{proof}

\repeatclaim{Corollary \ref{cor:thetaplow}}{\corthetaplow}
\begin{proof}
  Apply Corollary~\ref{cor:thetap} with $\eps := \gamma/\vartheta'(G)$.
\end{proof}

\repeatclaim{Corollary \ref{cor:thetaphigh}}{\corthetaphigh}
\begin{proof}
  Apply Corollary~\ref{cor:thetap} with $\eps := \gamma/\sqrt{n}$.
\end{proof}

\repeatclaim{Corollary \ref{cor:cara}}{\corcara}
\begin{proof}
  Let $B_i' :=
  \begin{bmatrix}
    \lambda_i B_i &\mathbf{0}\\
    \mathbf{0} &\lambda_i
  \end{bmatrix}$ for every $i\in[m]$ and
  $B' :=
  \begin{bmatrix}
    B& \mathbf{0}\\
    \mathbf{0} & 1
  \end{bmatrix}$, so that $B_i'\succeq 0$ and $B'= \sum_i B_i'$. By
  applying Theorem~\ref{thm:bssextension}, we obtain $y\in\Reals^m$ with
  $y\geq 0$ and $O(n/\eps^2)$ nonzero entries such that $B'\preceq
  \sum_i y_i B_i' \preceq (1+\eps)B'$ or, equivalently, $B\preceq \sum_y
  y_i\lambda_i B_i \preceq (1+\eps)B$ and $1\leq
  \sum_{i}y_i\lambda_i\leq 1+\eps$. Let $\mu\in\Reals^m$ be defined by
  $\mu_i := y_i\lambda_i/(\sum_iy_i\lambda_i)$. Then $\mu\geq 0$ and
  $\sum_{i}\mu_i = 1$, and
  \begin{equation*}
    (1-\eps)B
    \preceq
    \frac{B}{1+\eps}
    \preceq
    \frac{B}{\sum_{i}y_i\lambda_i}
    \preceq
    \sum_{i} \mu_i B_i
    \preceq
    \frac{1+\eps}{\sum_{i}y_i\lambda_i} B
    \preceq
    (1+\eps) B.
  \end{equation*}
  This completes the proof.
\end{proof}

\repeatclaim{Corollary \ref{cor:sim-sparsifier}}{\corsubgraph}

\begin{proof}
  For each edge $e \in E$, define $B_e := w_e
  \sqbrac[\big]{\Lcal_G(\chi^{e}) \oplus \bigoplus_{F \in \Fcal}
    \Lcal_F(\chi^{e}\!\!\restriction_{E(F)})}$, where $\chi^e$ denotes
  the characteristic vector of $\set{e}$ as a subset of~$E$. Now apply
  Theorem~\ref{thm:bssextension}.
\end{proof}

\section{The MMWUM}
\label{sec:mmwum-proofs}

In this section we provide some proofs about the MMWUM. These proofs
are due to Kale~\cite{Kale07}. Our set up and conclusions are slightly
different and we modified the proofs accordingly. We reproduce the
proofs here for the sake of completeness.

Theorem~\ref{thm:block-mmwum} can be viewed as a block-friendly
version of MMWUM. First we show the version with only one block. It is
basically the same as~\cite[Theorem~13 in Chapter 4]{Kale07}.

\begin{theorem}
  \label{thm:kale-mmwum}
  Let $T$ be a positive integer. Let $C, A_1,\dotsc, A_m \in
  \Sym$. Let $\eta > 0$ and $0 < \beta \leq 1/2$. For any given~$X \in
  \Sym$, consider the system
  \begin{equation}
    \label{kale-oracle}
    \displaystyle \sum_{i=1}^m y_i \iprod{A_i}{X}
    \geq \iprod{C}{X} - \eta \trace X,
    \quad\text{and}\quad y \in \Reals_+^m.
  \end{equation}
  Let $\set{\Pcal,\Ncal}$ be a partition of $[T]$, let $0 < \ell \leq
  \rho$, and let $W^{(t)}\in \Sym$ and $\ell^{(t)}
  \in \Reals$ for $t \in [T+1]$. Let $y^{(t)} \in \Reals^m$ for $t \in
  [T]$. Suppose the following properties hold:
  \begin{gather*}
    W^{(t+1)} = \exp\paren[\bigg]{-\frac{\beta}{\ell+\rho}
      \sum_{\tau=1}^t
      \sqbrac[\bigg]{\sum_{i=1}^m y_i^{(\tau)} A_i - C
        + \ell^{(\tau)} I} }, \qquad \forall t \in \set{0,\dotsc,T},
    \\
    y = y^{(t)} \text{ is a solution for~\eqref{kale-oracle} with } X
    = W^{(t)}, \qquad \forall t \in [T],
    \\
    \sum_{i=1}^m y_i^{(t)} A_i - C
    \in
    \begin{cases}
      [-\ell,\rho], & \text{if }t \in \Pcal,\\
      [-\rho,\ell], & \text{if }t \in \Ncal,
    \end{cases}
    \qquad \forall t \in [T],
    \\
    \ell^{(t)} = \ell, \quad \forall t \in \Pcal, \quad\text{and}\quad
    \ell^{(t)} = -\ell, \quad \forall t \in \Ncal.
  \end{gather*}
  Define $\yb := \frac{1}{T} \sum_{t=1}^T y^{(t)}$. Then
  \begin{equation}
    \label{eq:kale-mmwum-conclusion}
    \sum_{i=1}^m \yb_i A_i - C \succeq
    -\sqbrac[\Big]{\beta \ell + \frac{(\rho+\ell)\ln n}{T\beta}
      + (1+\beta)\eta}
    I.
  \end{equation}
\end{theorem}

The main tool for the proof of Theorem~\ref{thm:kale-mmwum} is the
following result:
\begin{theorem}[{Kale~\cite[Corollary~3 in Chapter~3]{Kale07}}]
  \label{thm:10}
  Let $0 < \beta \leq 1/2$. Let $T$ be a positive integer. Let
  $\set{\Pcal,\Ncal}$ be a partition of~$[T]$, and let $M^{(t)} \in
  \Sym$ for $t \in [T]$ and $W^{(t)} \in \Sym$ for $t \in [T+1]$ with
  the following properties:
  \begin{gather*}
    W^{(t+1)} = \exp\paren[\bigg]{-\beta \sum_{\tau=1}^t M^{(\tau)}}
    \quad\forall
    t=0,\dotsc,T, \\
    0 \preceq M^{(t)} \preceq I,\quad \forall t \in \Pcal,
    \qquad\text{and}\qquad
    -I \preceq M^{(t)} \preceq 0,\quad \forall t \in \Ncal,
  \end{gather*}
  Let \[ P^{(t)} := \frac{1}{\trace{W^{(t)}}} W^{(t)}, \qquad \forall
  t \in [T]. \] Then
  \begin{equation}
    \label{eq:3.3}
    (1-\beta)\sum_{t \in \Pcal} \iprod{M^{(t)}}{P^{(t)}} +
    (1+\beta)\sum_{t \in \Ncal} \iprod{M^{(t)}}{P^{(t)}}
    \leq
    \lambda_{\min} \paren[\bigg]{\sum_{t=1}^T M^{(t)}} +
    \frac{\ln n}{\beta}.
  \end{equation}
\end{theorem}
\begin{proof}
  Set $\Phi^{(t)} := \trace(W^{(t)})$ for $t \in [T+1]$. Put
  $\kaleeps_1 := 1 - e^{-\kaleeps}$ and $\kaleeps_2 :=
  e^{\kaleeps}-1$. Then, for any $t \in [T]$,
  \[
  \begin{split}
    \Phi^{(t+1)}
    & =
    \trace\paren{W^{(t+1)}}
    =
    \trace\paren[\Bigg]{\exp\paren[\bigg]{
        -\kaleeps\sum_{\tau=1}^t M^{(\tau)}}}
    \\
    & \leq
    \trace\paren[\Bigg]{
      \exp\paren[\bigg]{-\kaleeps\sum_{\tau=1}^{t-1} M^{(\tau)}}
      \exp\paren[\Big]{-\kaleeps M^{(t)}}}
    =
    \trace\paren[\Big]{W^{(t)} \exp(-\kaleeps M^{(t)})}
    \\
    & =
    \iprod{W^{(t)}}{\exp(-\kaleeps M^{(t)})},
  \end{split}
  \] where we have used Golden-Thompson's
  inequality~\eqref{eq:golden-thompson}.

  Using the fact that $e^x$ is convex, one can prove that
  \begin{align*}
    0 \preceq A \preceq I & \implies
    \exp(-\kaleeps A) \preceq I - \kaleeps_1 A, \\
    -I \preceq A \preceq 0 & \implies
    \exp(-\kaleeps A) \preceq I - \kaleeps_2 A.
  \end{align*}

  Suppose that $t \in \Pcal$. Then $\exp(-\kaleeps M^{(t)}) \preceq I
  - \kaleeps_1 M^{(t)}$, and since $W^{(t)} \succeq 0$, we get
  \[
  \begin{split}
    \Phi^{(t+1)}
    & \leq
    \iprod{W^{(t)}}{\exp(-\kaleeps M^{(t)})}
    \leq
    \iprod{W^{(t)}}{I-\kaleeps_1 M^{(t)}}
    \\
    & =
    \trace(W^{(t)})- \kaleeps_1 \iprod{W^{(t)}}{M^{(t)}}
    \\
    & =
    \trace(W^{(t)})- \trace(W^{(t)}) \kaleeps_1 \iprod{P^{(t)}}{M^{(t)}}
    \\
    & =
    \trace(W^{(t)})\sqbrac[\Big]{1- \kaleeps_1 \iprod{P^{(t)}}{M^{(t)}}}
    \\
    & =
    \Phi^{(t)} \sqbrac[\Big]{1- \kaleeps_1 \iprod{P^{(t)}}{M^{(t)}}}
    \\
    & \leq \Phi^{(t)} \exp\paren{-\kaleeps_1 \iprod{P^{(t)}}{M^{(t)}}}.
  \end{split}
  \] Similarly, if $t \in \Ncal$, then
  \[ \Phi^{(t+1)} \leq \Phi^{(t)} \exp\paren{-\kaleeps_2
    \iprod{P^{(t)}}{M^{(t)}}}. \]

  By induction on~$t$, and using $\Phi^{(1)} = \trace(I) = n$, we get
  \begin{equation*}
    \Phi^{(t+1)}
    \leq
    n \exp\paren[\bigg]{-\kaleeps_1 \sum_{\tau \in \Pcal \cap [t]}
      \iprod{M^{(\tau)}}{P^{(\tau)}}
      -\kaleeps_2 \sum_{\tau \in \Ncal \cap [t]}
      \iprod{M^{(\tau)}}{P^{(\tau)}}},
    \quad \forall t \in [T].
  \end{equation*}

  For every $A \in \Sym$, we have $\trace(\exp(A)) = \sum_{i=1}^n
  e^{\lambda_i} \geq e^{\lambda_j}$ for any $j \in [n]$, where
  $\lambda_1,\dotsc,\lambda_n$ are the eigenvalues of~$A$. Thus,
  \begin{equation*}
    \begin{split}
      \Phi^{(T+1)}
      & =
      \trace(W^{(T+1)}) = \trace\paren[\bigg]{\exp\paren[\bigg]{
          -\kaleeps \sum_{t=1}^T M^{(t)}}}
      \\
      & \geq
      \exp\paren[\bigg]{
        \lambda_{\max} \paren[\bigg]{-\kaleeps \sum_{t=1}^T M^{(t)}}}
      = \exp\paren[\bigg]{
        -\kaleeps \lambda_{\min} \paren[\bigg]{\sum_{t=1}^T M^{(t)}}}.
    \end{split}
  \end{equation*}
  Thus,
  \begin{equation*}
    \exp\sqbrac[\bigg]{ -\kaleeps \lambda_{\min} \paren[\bigg]{\sum_{t=1}^T
        M^{(t)}}} \leq
    n \exp\sqbrac[\bigg]{-\kaleeps_1 \sum_{t \in \Pcal}
      \iprod{M^{(t)}}{P^{(t)}} -\kaleeps_2 \sum_{t \in \Ncal}
      \iprod{M^{(t)}}{P^{(t)}}}.
  \end{equation*}

  By taking $\ln(\cdot)$ on both sides, we get
  \begin{equation*}
    -\kaleeps \lambda_{\min} \paren[\bigg]{\sum_{t=1}^T
      M^{(t)}}
    \leq
    \ln n - \sqbrac[\bigg]{\kaleeps_1 \sum_{t \in \Pcal}
      \iprod{M^{(t)}}{P^{(t)}} +\kaleeps_2 \sum_{t \in \Ncal}
      \iprod{M^{(t)}}{P^{(t)}}},
  \end{equation*}
  so
  \begin{equation*}
    \kaleeps_1 \sum_{t \in \Pcal}
    \iprod{M^{(t)}}{P^{(t)}} +\kaleeps_2 \sum_{t \in \Ncal}
    \iprod{M^{(t)}}{P^{(t)}}
    \leq
    \kaleeps \lambda_{\min} \paren[\bigg]{\sum_{t=1}^T M^{(t)}}
    +
    \ln n,
  \end{equation*}
  and
  \begin{equation*}
    \frac{\kaleeps_1}{\kaleeps} \sum_{t \in \Pcal}
    \iprod{M^{(t)}}{P^{(t)}}
    + \frac{\kaleeps_2}{\kaleeps} \sum_{t \in \Ncal}
    \iprod{M^{(t)}}{P^{(t)}}
    \leq
    \lambda_{\min} \paren[\bigg]{\sum_{t=1}^T M^{(t)}}
    +
    \frac{\ln n}{\kaleeps}.
  \end{equation*}

  Since $\sum_{t \in \Pcal} \iprod{M^{(t)}}{P^{(t)}} \geq 0$ and
  $\sum_{t \in \Ncal} \iprod{M^{(t)}}{P^{(t)}} \leq 0$, to
  prove~\eqref{eq:3.3} it suffices to show that $1-\kaleeps \leq
  \kaleeps_1/\kaleeps$ and $1+\kaleeps \geq \kaleeps_2/\kaleeps$. It
  is not hard to prove that
  \begin{equation*}
    \label{lemma:exp-bound}
    1-e^{-x} \geq x(1-x), \ \forall x \in [0,+\infty)
    \qquad
    \text{and}
    \qquad
    e^x-1 \leq x(1+x), \ \forall x \in [0,\myhalf]
  \end{equation*}
  So our choice of~$\kaleeps_1$ and~$\kaleeps_2$ ensures that
  $1-\kaleeps \leq \kaleeps_1/\kaleeps$ and $1+\kaleeps \geq
  \kaleeps_2/\kaleeps$.
\end{proof}

We can now show the proof of Theorem~\ref{thm:kale-mmwum}.
\begin{proof}[Proof of Theorem~\ref{thm:kale-mmwum}]
  Let $M^{(t)} := \frac{1}{\ell+\rho}\sqbrac[\Big]{\sum_{i=1}^m
    y_i^{(t)} A_i - C + \ell^{(t)} I}$ and $P^{(t)} :=
  W^{(t)}/\trace{W^{(t)}}$ for every $t$. For every $t \leq T$,
  using~\eqref{kale-oracle},
  \[
  \begin{split}
    \iprod{M^{(t)}}{P^{(t)}}
    & =
    \frac{1}{\ell+\rho}
    \sqbrac[\bigg]{
      \sum_{i=1}^m
      y_i^{(t)}
      \iprod{A_{i}}{P^{(t)}}
      - \iprod{C}{P^{(t)}}
      + \ell^{(t)}\iprod{I}{P^{(t)}}}
    \\
    & =
    \frac{1}{(\ell+\rho)\trace W^{(t)}}
    \sqbrac[\bigg]{
      \sum_{i=1}^m
      y_i^{(t)}
      \iprod{A_{i}}{W^{(t)}}
      - \iprod{C}{W^{(t)}}}
    + \frac{\ell^{(t)}}{\ell+\rho}
    \geq
    -\frac{\eta}{\ell+\rho}
    +
    \frac{\ell^{(t)}}{\ell+\rho},
  \end{split}
  \]
  since $y^{(t)}$ is a solution for~\eqref{kale-oracle} with $X :=
  W^{(t)}$. Thus, by~\eqref{eq:3.3},
  \begin{multline*}
    \sum_{t \in \Pcal} \frac{(1-\beta)(\ell^{(t)} - \eta)}{\ell+\rho}
    +
    \sum_{t \in \Ncal} \frac{(1+\beta)(\ell^{(t)}- \eta)}{\ell+\rho}
    \\
    \leq
    \frac{1}{\rho+\ell}
    \lambda_{\min}\paren[\Bigg]{
      \sum_{t=1}^T \sqbrac[\bigg]{
        \paren[\bigg]{\sum_{i=1}^m y_i^{(t)} A_{i}}
        - C  + \ell^{(t)}I}}
    + \frac{\ln n}{\beta}.
  \end{multline*}
  Multiply through by $\ell+\rho$ and move $\ell^{(t)}I$ out of
  $\lambda_{\min}(\cdot)$:
  \begin{multline*}
    \sum_{t \in \Pcal} {(1-\beta)\ell^{(t)}}
    +
    \sum_{t \in \Ncal} {(1+\beta)\ell^{(t)}}
    - T(1+\beta)\eta
    \\
    \leq
    \lambda_{\min}\paren[\Bigg]{
      \sum_{t=1}^T \sqbrac[\bigg]{
        \paren[\bigg]{\sum_{i=1}^m y_i^{(t)} A_{i}} - C}}
    + \paren[\bigg]{\sum_{t=1}^T \ell^{(t)}}
    + \frac{(\rho+\ell)\ln n}{\beta}.
  \end{multline*}
  Thus,
  \begin{equation*}
    \sum_{t \in \Pcal} {-\beta\ell^{(t)}}
    +
    \sum_{t \in \Ncal} {\beta\ell^{(t)}}
    \leq
    \lambda_{\min}\paren[\Bigg]{
      \sum_{t=1}^T \sqbrac[\bigg]{
        \paren[\bigg]{\sum_{i=1}^m y_i^{(t)} A_{i}}
        - C}}
    + \frac{(\rho+\ell)\ln n}{\beta}
    + T(1+\beta)\eta.
  \end{equation*}
  Next note that $\sum_{t \in \Pcal} -\ell^{(t)} + \sum_{t \in \Ncal}
  \ell^{(t)} = \sum_{t \in \Pcal} -\ell + \sum_{t \in \Ncal} - \ell =
  -T\ell$, so
  \begin{equation*}
    0
    \leq
    \lambda_{\min}\paren[\Bigg]{
      \sum_{t=1}^T \sqbrac[\bigg]{
        \paren[\bigg]{\sum_{i=1}^m y_i^{(t)} A_{i}}
        - C}} + \beta T\ell
    + \frac{(\rho+\ell)\ln n}{\beta}
    + T(1+\beta)\eta.
  \end{equation*}
  and
  \begin{equation*}
    0
    \leq
    \lambda_{\min}\paren[\Bigg]{ \frac{1}{T}
      \sum_{t=1}^T \sqbrac[\bigg]{\paren[\bigg]{
          \sum_{i=1}^m y_i^{(t)} A_{i}}
        - C}} + \beta \ell
    + \frac{(\rho+\ell)\ln n}{T\beta}
    + (1+\beta)\eta.
  \end{equation*}
  Thus,
  \[ \sum_{i=1}^m \yb_i A_{i} - C
  =
  \frac{1}{T}
  \sum_{t=1}^T \sqbrac[\bigg]{\paren[\bigg]{
      \sum_{i=1}^m y_i^{(t)} A_{i}} - C}
  \succeq -\sqbrac[\Big]{\beta \ell
    +\frac{(\rho+\ell)\ln n}{T \beta} + (1+\beta)\eta}I. \]
\end{proof}

Theorem~\ref{thm:block-mmwum} can be easily proved from
Theorem~\ref{thm:kale-mmwum}. First, we apply
Theorem~\ref{thm:kale-mmwum} separately for each block. In each
iteration, $y^{(t)}$ is a solution for~\eqref{kale-oracle} for all
blocks simultaneously, and so the conclusion
in~\eqref{eq:kale-mmwum-conclusion} holds for all blocks with same
$\bar y$. This new algorithm can be seen as equivalent to running~$K$
copies of MMWUM, each with different input data, with the caveat that
all copies run for the same number of iterations and the
vector~$y^{(t)}$ returned from the oracle is the same for all copies
at each iteration~$t$.

\section{Optimality of MMWUM Oracle}
\label{app:oracle_optimal}

\repeatclaim{Proposition \ref{prop:oracle_optimal}}{\proporopt}

\newcommand{\tensor}{\otimes}
\begin{proof}
Let $k=n/3$, let $I_k$ be the identity of size $k \times k$,
and let $e_j \in \Reals^k$ be the $j$th standard basis vector.
Let $\zeta = 3 \eta$ and define
$$
 X_1 = \Diag( 1, \zeta^3, \zeta ) \tensor I_k,
 \qquad
 X_2 = \Diag( 1, 1/\zeta^3, 1/\zeta ) \tensor I_k,
$$
where $\tensor$ denotes tensor product.
For $j = 1,\ldots,k$, define
$$
    v_{1,j} = [1/\sqrt{2}, -1/\sqrt{2}, 0] \tensor e_j,
    \qquad
    v_{2,j} = [1/\sqrt{2}, 1/\sqrt{2}, 0] \tensor e_j,
    \qquad
    v_{3,j} = [0, 0, 1] \tensor e_j.
$$
Let $B_{i,j} = v_{i,j} v_{i,j}^T$.
Note that $\sum_{i,j} B_{i,j} = I$.

The oracle cannot choose a matrix $B_{i,j}$ with $i \in \set{1,2}$,
since satisfying \eqref{eq:our-oracle-simple} would lead to a contradiction:
\begin{gather*}
    \frac{ \iprod{X_2}{B_i} }{ \trace(X_2) (1+\eta) }
    \leq \frac{1}{\alpha}
    \leq \frac{ \iprod{X_1}{B_i} }{ \trace(X_1) (1-\eta) }
\\\implies\quad
    1 + 3 \eta
    ~=~ 1 + \zeta
    ~<~ \frac{ \iprod{X_2}{B_i} / \trace X_2 }{ \iprod{X_1}{B_i} / \trace X_1 }
    ~\leq~ \frac{ 1 + \eta }{ 1 - \eta }
    ~<~ 1 + 3 \eta,
\end{gather*}
for sufficiently small $\eta$.

So the oracle must choose a matrix $B_{i,j}$ with $i=3$.
In this case,
\begin{gather*}
    \frac{ \trace B_{i,j} }{ \rho }
    \leq \frac{1}{\alpha}
    \leq \frac{ \iprod{ X_1 }{ B_{i,j} } }{ \trace(X_1) (1-\eta) }
\\\implies\quad
    \frac{ n }{ 9 \eta }
    ~=~ \frac{ n }{ 3 \zeta }
    ~\leq~ \frac{ (1 + \zeta^3 + \zeta) k }{ \zeta }
    ~=~ \frac{ \trace(B_{i,j}) \trace(X_1) }{ \iprod{ X_1 }{ B_{i,j} } }
    ~\leq~ \frac{ \rho }{ 1-\eta }.
\end{gather*}
This shows that $\rho = \Omega(n / \eta)$.
\end{proof}


\section{The positive semidefiniteness assumption}
\label{sec:psd-assumption}

\begin{proposition}
  For every positive integer $n$, there exist matrices $B_1, \dotsc, B_m
  \in \Sym$ with $m = \Omega(n^2)$ such that $B := \sum_i B_i$ is
  positive definite and with the following property: for every $\eps \in
  (0,1)$ and $y \in \Reals^m$ such that $(1-\eps)B \preceq \sum_i y_i
  B_i$, all entries of~$y$ are nonzero.
\end{proposition}
\begin{proof}
  Let $\Pcal := \setst{(i,j)}{i,j \in [n],\, i < j}$. For $(i,j) \in
  \Pcal$, let $E_{ij} := e_i e_j^T + e_j e_i^T$. Let $J$ denote the
  matrix of all ones. Then $2I + \sum_{(i,j) \in \Pcal} E_{ij} = I+J =:
  B \succ 0$. Let $\eps \in (0,1)$ and suppose that $(1-\eps)B \preceq
  2tI + \sum_{(i,j) \in \Pcal} z_{ij} E_{ij}$ for some $t \in \Reals$
  and $z \in \Reals^{\Pcal}$. By taking the inner product with $E_{ab}$
  on both sides, we see that $0 < 2(1-\eps) \leq z_{ab}$ for every
  $(a,b) \in \Pcal$. Similarly, we find that $0 < 2n(1-\eps) \leq 2nt$.
\end{proof}


\end{document}